\documentclass[11pt, a4paper,onecolumn]{belarticle4}
\pdfoutput=1
\usepackage{revsymb, amsmath, amsfonts, amssymb, enumerate, fullpage, amsthm, graphicx, braket, relsize, bbm, mathrsfs, mathtools}
\usepackage{graphicx,color}
\definecolor{darkbrown}{rgb}{0.4, 0.26, 0.13}
\definecolor{ao}{rgb}{0.0, 0.5, 0.0}
\definecolor{bleudefrance}{rgb}{0.19, 0.55, 0.91}

\usepackage{paralist}
\usepackage{subfigure}
\usepackage{soul}

\usepackage{tikz}
\usetikzlibrary{patterns}
\usetikzlibrary{calc,decorations.pathreplacing}
\usetikzlibrary{arrows,shapes}

\usepackage[linktocpage=true, colorlinks=true, linkcolor=blue, urlcolor=blue, citecolor=blue]{hyperref}

\newtheorem{theo}{Theorem}
\newtheorem{thm}[theo]{Theorem}
\newtheorem{prop}[theo]{Proposition}

\newtheorem{defn}[theo]{Definition}
\newtheorem{rem}[theo]{Remark}
\newtheorem{ex}[theo]{Example}

\newcommand{\cH}{\mathcal{H}}

\newcommand{\id}{\mathbb{I}}
\newcommand{\tr}[2]{\mathrm{tr}_{#2} \left\{ #1 \right\}}
\newcommand{\Tr}{\mathrm{tr}}

\newcommand{\cQ}{\mathcal{Q}}
\newcommand{\cNS}{\mathcal{NS}}
\newcommand{\X}{\mathbb{X}}

\newcommand{\A}{\mathbb{A}}

\newcommand{\As}{\boldsymbol{\Sigma}}

\newcommand{\T}{\mathbb{T}}

\newcommand{\blk}{\color{black}}
\definecolor{orangy}{RGB}{213,94,0}

\usepackage{cite} 

\usepackage{setspace}

\usepackage{tcolorbox}
\tcbuselibrary{skins,breakable}
\usetikzlibrary{shadings,shadows}

    {\endtcolorbox}

\usepackage[toc,page]{appendix}

\begin{document}

\title{Postquantum steering in scenarios with multiple characterised parties}

\author{Ana Bel\'en Sainz}
\affiliation{International Centre for Theory of Quantum Technologies, University of  Gda{\'n}sk, 80-309 Gda{\'n}sk, Poland \\ Basic Research Community for Physics e.V., Germany}

\date{}

\begin{abstract}
\vskip -1.7cm
The study of stronger-than-quantum phenomena (a.k.a., postquantum) has enabled a deeper understanding of the scope of quantum theory. Much is known about the case of correlations in Bell scenarios, where the so-called device-independent framework allowed us to explore its possibilities independently of the formalism of quantum theory quite earlier on. However, less is known about the phenomenon of Einstein-Podolsky-Rosen steering. Here, the so-called characterised parties are assumed to describe their systems locally through the quantum formalism, which  inconveniences  a theory-independent description. In addition, a renowned theorem by Gisin and Hughston, Josza and Wootters further hindered the discovery of the phenomenon. 
The study of postquantum steering, initiated about a decade ago, has been quite fruitful, including: the development of mathematical formalisms that frame the effect, resource theories that quantify it as a resource, and activation protocols that relate it to Bell correlations. However, all these results have a limitation in common: they apply to scenarios with only one quantum party.

In this work we articulate the concept of postquantum steering for scenarios with multiple quantum parties, bringing in the missing piece to the puzzle. We provide an algorithm to certify postquantumness, which in some cases also certifies quantumness. We also define a hierarchy of semidefinite programs that bounds the set of quantum assemblages from the outside. Moreover, we show that  the study of postquantum steering is fundamentally relevant since it is not just a mere mathematical curiosity allowed by the no-signalling principle, but it may arise within compositional theories beyond quantum theory. 
 Our work further discovers a peculiarity of steering: its theory-independent description  fundamentally prevents a direct connection with Bell nonlocality --  e.g., nonclassical Bell correlations do not imply nonclassical steering.
\end{abstract}

\vskip -3cm 

\maketitle

\begin{spacing}{0.3}
\vspace*{2cm}
\tableofcontents 
\end{spacing}

\section{Introduction}

Einstein-Podolsky-Rosen (EPR) steering is a puzzling nonlocal phenomenon featured of quantum theory \cite{schrodinger1935discussion,schrodinger36,wiseman2007steering}, first introduced by Schr\"odinger \cite{schrodinger1935discussion,schrodinger36}. In this setup, Alice and a distant Bob share a physical system, and the state of Bob's is seemingly remotely `steered' by Alice in a way that has no classical explanation, when she performs measurements on her share of the system. Steering has become popular within the quantum information community since it can be understood as a resource in situations where Alice's devices are uncharacterised or untrusted, enabling ``one-sided device independent'' implementations of information-theoretic tasks, such as quantum key distribution \cite{branciard2012one}, self-testing \cite{supic2016,gheorghiu2017}, randomness certification \cite{law2014quantum,passaro2015optimal}, and measurement incompatibility certification \cite{quintino2014joint,uola2014joint,cavalcanti2016quantitative}.

Quantum realisations of an EPR steering experiment may hence defy a classical explanation, and in doing so provide a quantum advantage in relevant tasks. The limits of this advantage and the scope of possibilities brought by quantum theory, however, are not yet fully understood, and an active area of research pertains to understanding the boundary between quantum steering and `steering allowed by other non-classical physical theories', hereon called `postquantum steering'. The study of this quantum boundary is similar in nature and motivation to the study of the boundary of quantum correlations in Bell experiments, where the mathematical formulation of Popescu-Rohrlich boxes \cite{popescu1994quantum} triggered a domino effect with impact on both foundational and applied topics \cite{brunner2014bell}. In the EPR steering scenario, however, such `quantum from the outside' exploration was only possible 
 after the discoveries of  
Ref.~\cite{sainz2015postquantum} due to a variety of challenges: (i) the EPR scenario has by definition a characterised quantum party, which make a theory-independent study complex, and (ii) the most studied EPR scenario is the bipartite one, and a celebrated theorem by Gisin \cite{gisin1989stochastic} and Hughston, Josza and Wootters \cite{hughston1993complete} implies that postquantum steering does not exist in such scenarios. How to mathematically articulate the concept of steering beyond that which quantum theory allows -- whilst still consistent with special relativity -- was firstly achieved in Ref.~\cite{sainz2015postquantum} for multipartite steering scenarios, and in Ref.~\cite{sainz2020bipartite} for generalised scenarios including bipartite ones. Since then, our understanding of postquantum steering has substantially progressed: we know that postquantum steering is a new phenomenon independent of postquantum Bell nonlocality \cite{sainz2015postquantum,sainz2020bipartite,sainz2018formalism}, postquantum steering provides a stringer-than-quantum advantage for some communication tasks \cite{cavalcanti2022post}, and postquantum steering may be activated to violate Bell-type inequalities in quantum networks \cite{sainz2025activation,zjawin2024activation}. New technical tools to study postquantum steering have been developed \cite{hoban2025hierarchy}, which enabled some of the above-mentioned studies. What is moreover fascinating, is that postquantum steering may go beyond being a mere mathematical curiosity, since some compositional foil theories feature the phenomenon \cite{cavalcanti2022post,cavalcanti2024every}. 

The study of postquantum steering, however, has so far featured a limitation: only scenarios with one characterised (quantum) party have been considered, such as those depicted in Fig.~\ref{fig:steebip}. In this paper we extend the exploration of postquantum steering to scenarios with multiple characterised parties, completing the puzzle of postquantumness for EPR scenarios. We start by identifying and formalising what the relevant object of study (assemblage) is in the study of this phenomenon -- this is the analogue of the `correlations' in a Bell scenario, or `assemblages' in EPR scenarios with only one characterised party. Then, we illustrate ways in which EPR steering may defy a quantum explanation in our scenario of interest. We develop techniques to certify postquantum steering, which we implement for our numerical exploration.  We then show that postquantum steering in these scenarios is also a phenomenon independent of Bell-type nonlocality, and that it may arise in compositionally-sound foil theories, making the study of postquantum steering relevant beyond as just a mathematical curiosity. We then discuss the robustness of the postquantum steering proofs presented in the paper.   
We further introduce a hierarchy of semidefinite programs that outer-approximates the set of quantumly-realisable assemblages.
 Throughout the manuscript we show that the relation between steering and Bell nonlocality is not a strict hierarchy (they are indeed incomparable phenomena) when one endorses a theory-independent operational definition of assemblage.

\begin{figure}
\begin{center}
\subfigure[\quad Traditional]{
		\begin{tikzpicture}[scale=0.5]
		\node at (-2,1.3) {Alice};
		\shade[draw, thick, ,rounded corners, inner color=white,outer color=gray!50!white] (-1.7,-0.3) rectangle (-2.3,0.3) ;
		\draw[thick, ->] (-2.5,0.5) to [out=180, in=90] (-2,0.3);
		\node at (-2.8,0.5) {$x$};
		\draw[thick, -<] (-2,-0.3) to [out=-90, in=180] (-2.5,-0.5);
		\node at (-2.9,-0.5) {$a$};

		\node at (2,1.3) {Bob};
		\shade[draw, thick, ,rounded corners, inner color=white,outer color=gray!50!white] (1.7,-0.3) rectangle (2.3,0.3) ;		
		\draw[thick, ->] (2,-0.3) -- (2,-0.7);
		\node at (2,-1) {$\sigma_{a|x}$};		
		
		\node at (0,0) {s};
		\draw[thick, dashed, color=gray!70!white] (0,0) circle (0.3cm);
		\draw[thick, dashed, color=gray!70!white, ->] (-0.3,0) -- (-1.6,0);
		\draw[thick, dashed, color=gray!70!white, ->] (0.3,0) -- (1.6,0);		
	    \end{tikzpicture}	}
	    	    \hspace{0.5cm}
\subfigure[\quad Multipartite-uncharacterised]{	    
	    \begin{tikzpicture}[scale=0.5]
		\node at (-2,1.3) {Alice$_1$};
		\shade[draw, thick, ,rounded corners, inner color=white,outer color=gray!50!white] (-1.7,-0.3) rectangle (-2.3,0.3) ;
		\draw[thick, ->] (-2.5,0.5) to [out=180, in=90] (-2,0.3);
		\node at (-2.8,0.5) {$x_1$};
		\draw[thick, -<] (-2,-0.3) to [out=-90, in=180] (-2.5,-0.5);
		\node at (-3,-0.5) {$a_1$};

		\node at (6,1.3) {Alice$_2$};
		\shade[draw, thick, ,rounded corners, inner color=white,outer color=gray!50!white] (5.7,-0.3) rectangle (6.3,0.3) ;
		\draw[thick, ->] (6.5,0.5) to [out=180, in=90] (6,0.3);
		\node at (6.8,0.5) {$x_2$};
		\draw[thick, ->] (6,-0.3) to [out=-90, in=180] (6.5,-0.5);
		\node at (7,-0.5) {$a_2$};

		\node at (2,1.3) {Bob};
		\shade[draw, thick, ,rounded corners, inner color=white,outer color=gray!50!white] (1.7,-0.3) rectangle (2.3,0.3) ;		
		\node at (2,-1) {$\sigma_{a_1a_2|x_1x_2}$};
		\draw[thick, ->] (2,-0.3) -- (2,-0.7);
		
		\node at (0,0) {s};
		\draw[thick, dashed, color=gray!70!white] (0,0) circle (0.3cm);
		\draw[thick, dashed, color=gray!70!white, ->] (-0.3,0) -- (-1.6,0);
		\draw[thick, dashed, color=gray!70!white, ->] (45:0.3) to [out=20, in=160] (5.6,0);		
		\draw[thick, dashed, color=gray!70!white, ->] (0.3,0) -- (1.6,0);

	    \end{tikzpicture}}  
	    \hspace{0.5cm}
\subfigure[\quad Bob with input]{	    
	    \begin{tikzpicture}[scale=0.5]
		\node at (-2,1.3) {Alice};
		\shade[draw, thick, ,rounded corners, inner color=white,outer color=gray!50!white] (-1.7,-0.3) rectangle (-2.3,0.3) ;
		\draw[thick, ->] (-2.5,0.5) to [out=180, in=90] (-2,0.3);
		\node at (-2.8,0.5) {$x$};
		\draw[thick, -<] (-2,-0.3) to [out=-90, in=180] (-2.5,-0.5);
		\node at (-2.9,-0.5) {$a$};

		\node at (2,1.3) {Bob};
		\shade[draw, thick, ,rounded corners, inner color=white,outer color=gray!50!white] (1.7,-0.3) rectangle (2.3,0.3) ;		
		\draw[thick, ->] (2.5,0.5) to [out=180, in=90] (2,0.3);
		\node at (2.7,0.5) {$y$};
		\draw[thick, ->] (2,-0.3) -- (2,-0.7);
		\node at (2,-1) {$\sigma_{a|xy}$};
		
		\node at (0,0) {s};
		\draw[thick, dashed, color=gray!70!white] (0,0) circle (0.3cm);
		\draw[thick, dashed, color=gray!70!white, ->] (-0.3,0) -- (-1.6,0);
		\draw[thick, dashed, color=gray!70!white, ->] (0.3,0) -- (1.6,0);		
		\draw[thick, dashed, color=gray!70!white, ->] (0.3,0) -- (1.6,0);		

	    \end{tikzpicture}}	    
\end{center}
\caption{\textbf{Different EPR setups:} (a) traditional scenario: Alice makes a measurement, steering the state of Bob; (b)  multipartite EPR scenario with one Bob: two parties make independent measurements on their share of a system and steer the state of Bob (c) Bob-with-input (BWI) scenario: Bob has an input, and influences his state preparation, by performing some operation on it.}
\label{fig:steebip}
\end{figure}
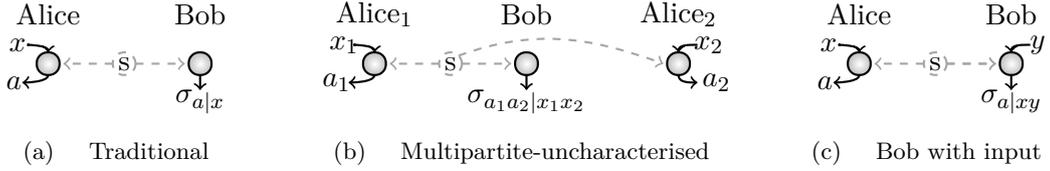

\section{The scenario: definitions}

The EPR scenario we explore is that where we have at least two characterised parties, such as the one depicted in Fig.~\ref{fig:steeint}. The simplest such scenario is that with only one Alice and two Bobs (see Fig.~\ref{fig:steeint}), and throughout this manuscript we will sometimes focus the discussion on this case for simplicity.

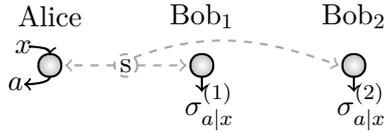
\begin{figure}
\begin{center}
	    
	    \begin{tikzpicture}[scale=0.5]
		\node at (-2,1.3) {Alice};
		\shade[draw, thick, ,rounded corners, inner color=white,outer color=gray!50!white] (-1.7,-0.3) rectangle (-2.3,0.3) ;
		\draw[thick, ->] (-2.5,0.5) to [out=180, in=90] (-2,0.3);
		\node at (-2.7,0.5) {$x$};
		\draw[thick, -<] (-2,-0.3) to [out=-90, in=180] (-2.5,-0.5);
		\node at (-2.9,-0.5) {$a$};

		\node at (6,1.3) {Bob${}_2$};
		\shade[draw, thick, ,rounded corners, inner color=white,outer color=gray!50!white] (5.7,-0.3) rectangle (6.3,0.3) ;
		\node at (6.2,-1.1) {$\sigma^{(2)}_{a|x}$};
		\draw[thick, ->] (6,-0.3) -- (6,-0.7);

		\node at (2,1.3) {Bob${}_1$};
		\shade[draw, thick, ,rounded corners, inner color=white,outer color=gray!50!white] (1.7,-0.3) rectangle (2.3,0.3) ;		
		
		\node at (2.2,-1.1) {$\sigma^{(1)}_{a|x}$};
		\draw[thick, ->] (2,-0.3) -- (2,-0.7);

		\node at (0,0) {s};
		\draw[thick, dashed, color=gray!70!white] (0,0) circle (0.3cm);
		\draw[thick, dashed, color=gray!70!white, ->] (-0.3,0) -- (-1.6,0);
		\draw[thick, dashed, color=gray!70!white, ->] (45:0.3) to [out=20, in=160] (5.6,0);		
		\draw[thick, dashed, color=gray!70!white, ->] (0.3,0) -- (1.6,0);		

	    \end{tikzpicture}  
\end{center}
\caption{\textbf{EPR setup of interest:} EPR steering scenario where one uncharacterised party (Alice) steers the states of the subsystems held by two characterised parties (Bobs). 
}
\label{fig:steeint}
\end{figure}

The first question is how the steering scenario should be defined from an operational point of view. For example, in the case of a Bell scenario, this characterisation is done by saying what the number of parties is, what the number of measurements that each party has access to is, and what the number of outcomes of each measurement is. For this steering scenario, let $n$ denote the number of uncharacterised parties (Alices). For each Alice ($j=1:n$), then, denote by $\X_j$ the set of labels for her measurement choices, and by $\A^j_k$ the set of labels for the outcomes of her $k$-th measurement. In addition, let us denote by $N$ the number of characterised parties in the scenario (Bobs). For each Bob, let $d_\ell$ be the dimension of the Hilbert space where the state of his system is described. Then, a natural set of parameters to specify this steering scenario arises.
\begin{defn}\textbf{Parameters to specify our multipartite steering scenario of interest.--}\\
A multipartite steering scenario with many characterised parties (Bobs)  is specified by the following parameters:
\begin{align}\label{eq:tuple}
\T = \left(n, \{|\X_j|\}_{j=1:n},\{|\A^j_k|\}_{k=1:|\X_j|,j=1:n},N,\{d_\ell\}_{\ell=1:N} \right)\,,
\end{align}
where $n$ denotes the number of Alices, $\X_j$ denotes the set of labels for the measurement choices of the $j$-th Alice, $\A^j_k$ denotes the set of labels for the outcomes of the $k$-th measurement of the $j$-th Alice, $N$ denotes the number of Bobs, and $d_\ell$ denotes the dimension of the Hilbert space where the system of the $\ell$-th Bob is described. 
\end{defn}

Similarly to the case of Bell scenarios, whenever $|\A^j_k|=|\A^{j'}_{k'}| \quad \forall \, j,j',k,k'$ and $|\X_j| = \X_{j'}| \quad \forall \, j,j'$, we will specify the tuple of Eq.~\eqref{eq:tuple} simply by
\begin{align}\label{eq:tuples}
\T = \left(n, |\X|,|\A|,N,\{d_\ell\}_{\ell=1:N} \right)\,.
\end{align}
In addition, whenever the Hilbert space dimension of the characterised parties is the same, i.e., $d_\ell = d_{\ell'} \quad \forall \, \ell,\ell'$, we will further simply the specification of the tuple of Eq.~\eqref{eq:tuple} and write
\begin{align}\label{eq:tupless}
\T = \left(n, \{|\X_j|\}_{j=1:n},\{|\A^j_k|\}_{k=1:|\X_j|,j=1:n},N,d_\ell \right)\,.
\end{align}

Now that we have a way to specify the EPR scenario, the next question is how to define the object of interest from an operational viewpoint, i.e., \textbf{what is an assemblage} in this scenario. We will denote an assemblage by $\As_\T$. Our fundamental grounds here are that we have only access to the \textit{local} information of the parties (both Alices and Bobs). With the classical information from the Alices we can infer a joint conditional probability distribution for their output statistics, like traditionally done in Bell experiments. However, the same is not true for the data observed by the Bobs: we do not want to make any assumptions on how this information is processed to generate a global state of a joint $N$-partite physical system\footnote{For instance, we cannot define the assemblage elements to involve the tensor product of the individual description of the states of the subsystems (i.e., $\sigma^{(1)}_{a|x} \otimes \sigma^{(2)}_{a|x}$).}. Hence, we need to describe $\As_\T$ in terms of the statistical data observed by the Alices together with the local quantum states describing the system of each Bob. Therefore, our object of study is specified by the following collection of tuples: 

\begin{defn}\label{def:assem}\textbf{Assemblage.--}\\
An assemblage $\As_\T$ in the steering scenario specified by $\T$ is given by the following assemblage elements: 
\begin{align}
\As_\T = \left\{ \left(p(a_1 \ldots a_n|x_1 \ldots x_n), \rho^{(1)}_{a_1 \ldots a_n|x_1 \ldots x_n} ,\ldots, \rho^{(N)}_{a_1 \ldots a_n|x_1 \ldots x_n}\right)  \right\}_{a_1 \ldots a_n,x_1 \ldots x_n}\,,
\end{align}
where $p(a_1 \ldots a_n|x_1 \ldots x_n)$ is the probability that the Alices locally obtain outcomes $a_1 \ldots a_n$ when locally performing the measurements labelled by $x_1 \ldots x_n$, and $\rho^{(j)}_{a_1 \ldots a_n|x_1 \ldots x_n}$ is the normalised quantum state that describes  the subsystem of the $j$-th characterised party when the Alices perform measurements $x_1 \ldots x_n$ and obtain outcomes $a_1 \ldots a_n$. 
\end{defn}
For simplicity in the notation, let us denote $\mathbf{a} := a_1 \ldots a_n$ and $\mathbf{x} := x_1 \ldots x_n$. 

\bigskip 

One could also mathematically incorporate the correlations $p(\mathbf{a}|\mathbf{x})$ as the normalisation of the states of each local system, in the following way: 
\begin{align}\label{eq:asssub}
\sigma^{(j)}_{\mathbf{a}|\mathbf{x}} = p(\mathbf{a}|\mathbf{x}) \, \rho^{(j)}_{\mathbf{a}|\mathbf{x}} \,,
\end{align}
where now $\sigma^{(j)}_{\mathbf{a}|\mathbf{x}}$ denotes the (possibly subnormalised) state of the subsystem held by the $j$-th Bob. With this convention, an assemblage can equivalently be specified as 
\begin{align}
\As_\T = \left\{ \left( \sigma^{(1)}_{\mathbf{a}|\mathbf{x}} ,\ldots, \sigma^{(N)}_{\mathbf{a}|\mathbf{x}}\right)  \right\}_{\mathbf{a},\mathbf{x}}\,,
\end{align}
where $p(\mathbf{a}|\mathbf{x}) = \tr{\sigma^{(k)}_{\mathbf{a}|\mathbf{x}}}{}$ for all $k$.

Notice that, from a fundamental perspective, this operational choice of definition of assemblage $\As_\T$ contains no information on how the Bobs are correlated. Hence, were the Bobs to measure their quantum systems, all one can infer are the correlations $p(\mathbf{a}b_k|\mathbf{x}y_k)$ among the Alices and each Bob.

\bigskip 

A remark is in order about this choice for the definition of assemblage. Steering with more than one characterised party has been explored in the context of quantum theory featuring different types of entanglement \cite{cavalcanti2015detection,uola2020quantum}. The study then starts from a quantum underpinning of the experiments and asks for alternative explanations\footnote{Various works in the literature tackle  this scenario from the viewpoint of entanglement certification, so sometimes the question is whether a quantum system prepared in a `fully separable state' can reproduce the assemblage (which is equivalent to having the parties construct the assemblage by performing possibly-quantum local operations correlated via a classical common cause), whereas in some other cases people just ask whether genuinely-multipartite entanglement is necessary for reproducing the assemblage.}, which may sometimes be classical. Importantly, the collection of density matrices under study (what is there defined as an assemblage) comprise the joint quantum system of all the Bobs. That is, an assemblage in these multipartite quantum steering scenarios with many Bobs contain information on how these Bobs are correlated. This is in great contrast to what we do in this paper, where we only leverage the local information that the Bobs have access to.  The idea here is that we do not want to assume that quantum theory underpins our experiment, but we want to acknowledge that locally the Bobs can characterise their systems using the quantum formalism (after all, quantum is a valid theory that should still hold  in some regime of applicability even if a superseding theory emerges). Hence the question is: what opportunities arise when we allow for the global physics to be underpinned by some other exotic theory that locally looks quantum in this experiment.  
In Ap.~\ref{app:compare} we briefly discuss classical vs.~quantum steering in our scenario, and how this compares with the state of the art.

\bigskip

Since in this paper we are interested in the interface between quantum and postquantum theories, now we need to specify  two things: what we mean for an \textbf{assemblage to be non-signalling} and \textbf{to admit of a quantum realisation in this EPR scenario}. Let us start with the former.

Recall that an assemblage may be specified by ${\As_\T = \left\{ \left( p(\mathbf{a}|\mathbf{x}), \rho^{(1)}_{\mathbf{a}|\mathbf{x}} ,\ldots, \rho^{(N)}_{\mathbf{a}|\mathbf{x}}\right)  \right\}_{\mathbf{a},\mathbf{x}}}$. Hence, a necessary feature for the assemblage to be non-signalling is that the correlations $p(\mathbf{a}|\mathbf{x})$ among the Alices be non-signalling. The Bobs don't perform any measurements on their systems, so they cannot signal to any other party. What remains to check is what conditions prevent any collection of Alices to signal to any collection of Bobs. So let us begin the discussion with signalling across the bipartition Alices vs Bobs: here, when tracing out the Alices, each Bob ends up with a system prepared in the normalised reduced state 
\begin{align}
\sigma^{(k)}_{\mathbf{R},\mathbf{x}} = \sum_{\mathbf{a}} p(\mathbf{a}|\mathbf{x}) \, \rho^{(k)}_{\mathbf{a}|\mathbf{x}}\,, \forall \, k = 1:N\,.
\end{align}
Now, for the Alices to not signal any single Bob, one must have that each reduced state $\sigma^{(k)}_{\mathbf{R},\mathbf{x}}$ do not depend on $\mathbf{x}$. 
Notice that, in this case, if a subset of the Bobs get together, they will not  be able to infer any particular value of $\mathbf{x}$ just by comparing the reduced states of their quantum systems. 
Hence, the condition that $\sigma^{(k)}_{\mathbf{R},\mathbf{x}}$ be independent of $\mathbf{x}$ is sufficient to prevent the Alices to signal to any collection of the Bobs too. One can now ask the question of a subset of the Alices signalling to a (possibly complete) subset of the Bobs. Let us denote by $\mathbf{S} = \{i_1,\ldots,i_r\}$ the subset of $r \leq n$ uncharacterised parties $\{A_{i_1},\ldots,A_{i_r}\}$. If we take the marginal over the parties not in the set $\mathbf{S}$, the Bobs end each with a system prepared in the state 
\begin{align}\label{eq:8}
\sigma^{(k)}_{\mathbf{a}_\mathbf{S},\mathbf{x}} = \sum_{a_j \in \A^j_{x_j} \,|\, j \not\in \mathbf{S}}  p(\mathbf{a}|\mathbf{x}) \, \rho^{(k)}_{\mathbf{a}|\mathbf{x}}\,,
\end{align}
where $\mathbf{a}_\mathbf{S}$ denotes the measurement outcomes of the subset of Alices in $\mathbf{S}$. 
Now, for the marginalised Alices (i.e., those not in $\mathbf{S}$) to not be able to signal to each single Bob, we need that $\sigma^{(k)}_{\mathbf{a}_\mathbf{S},\mathbf{x}}$ does not depend on their choice of measurement. That is, we need that $\sigma^{(k)}_{\mathbf{a}_\mathbf{S},\mathbf{x}} = \sigma^{(k)}_{\mathbf{a}_\mathbf{S},\mathbf{x}_\mathbf{S}}$ for all $\mathbf{x}$ and for all $k=1,\ldots,N$, where $\mathbf{x}_\mathbf{S}$ denotes the measurement choices of the subset of Alices in $\mathbf{S}$. Notice first that, when $\mathbf{S} = \emptyset$, then we recover the condition we discussed before for $\sigma^{(k)}_{\mathbf{R},\mathbf{x}}$. Second, notice that this condition is also sufficient to guarantee that the marginalised Alices cannot signal to any subset of the Bobs should the latter come together  and compare the reduced states of their quantum systems. 
 Hence, we have a complete set of conditions to ensure that an assemblage is non-signalling. 
These conditions actually admit of a compact form if one recalls the definition of no-signalling assemblages for EPR scenarios with a single Bob \cite{sainz2018formalism,sainz2015postquantum}. Indeed, consider the steering scenario defined by all the $n$ Alices and a single Bob (say, the the $k$-th one). The assemblage $\left\{ \left( p(\mathbf{a}|\mathbf{x}), \rho^{(k)}_{\mathbf{a}|\mathbf{x}}\right)  \right\}_{\mathbf{a},\mathbf{x}}$ in this $n+1$ EPR scenario with a single-Bob is non-signalling precisely if and only if Eq.~\eqref{eq:8} is satisfied. This leads to the following definition. 

\begin{defn}\textbf{Weakly-Non-signalling assemblage.--}\\
Consider an EPR scenario specified by $\T$, and an assemblage $\As_\T = \left\{ \left( \sigma^{(1)}_{\mathbf{a}|\mathbf{x}} ,\ldots, \sigma^{(N)}_{\mathbf{a}|\mathbf{x}}\right)  \right\}_{\mathbf{a},\mathbf{x}}$.  

\noindent $\As_\T $ is non-signalling if and only if the following are satisfied: 
\begin{compactenum}
\item The correlations $p(\mathbf{a}|\mathbf{x}) = \tr{\sigma^{(k)}_{\mathbf{a}|\mathbf{x}}}{}$ are non-signalling. 
\item The single-Bob assemblage $\As_\T^k = \left\{\sigma^{(k)}_{\mathbf{a}|\mathbf{x}} \right\}_{\mathbf{a},\mathbf{x}}$ is non-signalling for each $k=1,\ldots,N$. 
\end{compactenum}
We denote by $\cNS_\T$ the set of all assemblages that are  weakly-non-signalling  for a given scenario $\T$. 
\end{defn}

\bigskip

You may notice that we have used the term \textit{weakly-non-signalling} in the definition above. This is because, as discussed previously, this only pertains to the Bobs trying to infer $\mathbf{x}$ by merely comparing the reduced states of their systems. But one could take a step further and let the Bobs make measurements on their quantum systems (possibly jointly) and then make inferences on  $\mathbf{x}$ from the collected joint statistics. Preventing this from happening would entail a stronger form of non-signalling. Now, how to impose this in full generality is a tricky problem, since which measurements can the Bobs perform globally depends on the actual theory underpinning the experiment. For instance, if quantum theory holds globally, then the Bobs can perform entangling Bell measurements on their systems, whereas if the experiment is underpinned by Witworld then they can only measure their systems with separable quantum measurements. To explore this stronger form of non-signalling constraints, we find it useful to revise the notion of steering in foil theories \cite{jenvcova2022assemblages}.
In this type of steering, explored for a single Bob, one assumes a specific theory as underpinning the experiment (classical, quantum, or some foil theory). Bob then is though of as in possession of a system of a type allowed by that theory, and we study the collection of (possibly subnormalised) states of those (possibly not quantum) systems. Hence, the assemblages are not necessarily given by states of quantum systems. To relate this to our case of study, we can think of an assemblage ${\As_\T = \left\{ \left( \sigma^{(1)}_{\mathbf{a}|\mathbf{x}} ,\ldots, \sigma^{(N)}_{\mathbf{a}|\mathbf{x}}\right)  \right\}_{\mathbf{a},\mathbf{x}}}$ as being underpinned by an assemblage ${\As_\T^G = \left\{ \tau_{\mathbf{a}|\mathbf{x}} \right\}_{\mathbf{a},\mathbf{x}}}$ of states of a system from a theory $G$ -- where each $\tau_{\mathbf{a}|\mathbf{x}}$ is a valid state of the $N$-partite system in $G$ -- when $\sigma^{(k)}_{\mathbf{a}|\mathbf{x}} = \mathrm{trace}_{\neg k} \left[ \tau_{\mathbf{a}|\mathbf{x}} \right] $, where  by $\mathrm{trace}_{\neg k}$ we mean the discard operation in $G$ of all subsystems apart from the $k$-th one. Should this theory actually underpin our physical experiment, then we can conclude that our assemblage $\As_\T$ does not allow the Alices to signal to any (possibly complete) subset of the Bobs. However, if such a theory exists, this doesn't rule out the possibility that the experiment was instead carried out using a different theory that allows for signalling among the parties when the Bobs measure.  This conundrum makes us take a step back and ask: why do we actually care about this stronger form of non-signalling? 
Indeed, at the end of the day all we want is to make the following argument: if an assemblage defies a quantum explanation, can this not be just a consequence of signalling. Given this, we settle for the following conditions for strongly-non-signalling assemblages.

\begin{defn}\textbf{Strongly-Non-signalling assemblage.--}\\
Consider an EPR scenario specified by $\T$, and an assemblage $\As_\T = \left\{ \left( \sigma^{(1)}_{\mathbf{a}|\mathbf{x}} ,\ldots, \sigma^{(N)}_{\mathbf{a}|\mathbf{x}}\right)  \right\}_{\mathbf{a},\mathbf{x}}$.  

\noindent $\As_\T $ is strongly-non-signalling if and only if the following are satisfied: 
\begin{compactenum}
\item $\As_\T$ is weakly-non-signalling,
\item There exists a theory $G$ that includes local quantum systems, an $N$-partite system $s_N$ whose subsystems are quantum (with the $k$-th system belonging to $\cH_{d_k}$), and an assemblage ${\As_\T^G = \left\{ \tau_{\mathbf{a}|\mathbf{x}} \right\}_{\mathbf{a},\mathbf{x}}}$ of states of $s_N$, such that $\sigma^{(k)}_{\mathbf{a}|\mathbf{x}} = \mathrm{trace}_{\neg k} \left[ \tau_{\mathbf{a}|\mathbf{x}} \right] $ for all $k=1,\ldots,N$,
\item $\sum_{\mathbf{a}} \tau_{\mathbf{a}|\mathbf{x}} = \sum_{\mathbf{a}} \tau_{\mathbf{a}|\mathbf{x'}} $ for all $\mathbf{x},\mathbf{x'}$.
\end{compactenum}
\end{defn}

Notice therefore that, if an strongly-non-signalling assemblage does not admit of a quantum realisation (defined below), then there exists an explanation of this phenomenon that does not require signalling. 

\bigskip

We can now discuss the notion of an assemblage being quantumly realisable.

\begin{defn}\textbf{Quantum realisation of an assemblage.--}\\
Given an EPR scenario specified by $\T$, an assemblage $\As_\T$ admits of a quantum realisation if there exists a Hilbert space $\cH_{A_j}$ for each $j=1:n$, a quantum state $\rho$ in the Hilbert space $\left(\otimes_{j=1:n}\cH_{A_j} \right) \otimes \left(\otimes_{\ell=1:N}\cH_{d_\ell}\right)$, and measurements $\{\{M^{(j)}_{a_j|x_j}\}_{a_j\in\A^j_x}\}_{x \in \X_j}$ for each Alice $j=1:n$, such that:
\begin{align}
p(\mathbf{a}|\mathbf{x}) &= \Tr\left\{ \left(\otimes_{j=1:n}M^{(j)}_{a_j|x_j} \right) \otimes \id_{d_1 \ldots d_N} \, \rho  \right\} \,,\\ 
\sigma^{(k)}_{\mathbf{a}|\mathbf{x}} &= \Tr_{\cH_{\neg k}} \left\{\left(\otimes_{j=1:n}M^{(j)}_{a_j|x_j} \right) \otimes \id_{d_1 \ldots d_N} \, \rho \right\}\,,
\end{align}
where $\cH_{\neg k} := \left(\otimes_{j=1:n}\cH_{A_j} \right) \otimes \left(\otimes_{\ell=1:N\,,\ell \neq k}\cH_{d_\ell}\right)$. 

 We denote by $\cQ_\T$ the set of all assemblages that admit of a quantum realisation for a given scenario $\T$.  
\end{defn} 

 Notice that the states and measurements that provide a quantum realisation for an assemblage $\As_\T$, guarantee that the assemblage is  strongly-non-signalling  by construction. Hence, $\cQ_\T \subseteq \cNS_\T$.  

As a remark, notice that a quantumly-realisable assemblage always features quantum correlations among the Alices:

\begin{rem}
If the assemblage is quantumply-realisable, then the correlations  ${p(\mathbf{a}|\mathbf{x}) = \tr{\sigma^{(k)}_{\mathbf{a}|\mathbf{x}}}{}}$   observed between the measurement outcomes of the Alices admit of a quantum realisation. 
\end{rem}

\section{Steering beyond quantum theory: proof of principle}

We can now explore the postquantum realm of assemblages in this multipartite EPR scenario. The following is an example of postquantum steering, whose simplicity serves as a friendly proof-of-principle.

\begin{ex} \textbf{Postquantum steering example.--}\\
A natural example of postquantum steering arises in the scenario $\T = (n=2,|x|=2,|a|=2,N=2,d_1=2,d_2=2)$. We take the labels of Alices' measurements and outcomes to be $\X_1 = \X_2 = \A^j_k = \{0,1\}$. Define the assemblage $\As_\T$ as that with elements:
\begin{align}
\left(p(a_1a_2|x_1x_2),\frac{\id_2}{2},\frac{\id_2}{2} \right)\,\\
\end{align}
where $p(a_1a_2|x_1x_2) = \frac{1}{2}\delta_{a_1 \oplus a_2 = x_1x_2}$ are Popescu-Rohrlich correlations \cite{popescu1994quantum}. 

The assemblage $\As_\T$ is postquantum. This follows from the previous Remark and noticing that the correlations between the Alices are not quantum (they are a PR-box) \cite{popescu1994quantum}.  Notice moreover that the assemblage is non-signalling, since Popescu-Rohrlich correlations are themselves non-signalling.  In addition, this assemblage is strongly-non-signalling  since it may arise as a common-cause process within the toy theory Witworld \cite{cavalcanti2022post}. 
\end{ex}

Although this example shows how postquantumness can arise here, it is not the most enlightening, since the Bobs play no significant role in the phenomenon. What we need is an example where the specific states of the systems held by the Bobs play a meaningful role. 
The challenge here is that we do not have any information about any way that the systems of the characterised parties correlate; we only have partial information about the physical situation given by the local marginals. In the next section we discuss tools to explore the scenario in a holistic way, and present meaningful examples. 

\section{Postquantumness certification}\label{se:ntc}

We will now discuss an approach to certifying postquantumness of assemblages in scenarios with two or more Bobs. The technique focuses on finding global quantum states that have the Bobs' states as marginals. 

\bigskip 

For clarity in the presentation let us focus on the case where the scenario $\T$ has one Alice ($n=1$) and two Bobs ($N=2$) -- the generalisation to arbitrary $\T$ will be made explicit whenever it is not straightforward. An assemblage in this scenario is given by 
\begin{align}
\As_\T = \{ (p(a|x), \rho^{(1)}_{a|x}, \rho^{(2)}_{a|x})  \}_{a|x}\,.
\end{align}
In order for this to admit of a quantum realisation, we need the following: 
\begin{compactitem}
\item[(i)] For each  $(a,x)$, $ \rho^{(1)}_{a|x}$ and $\rho^{(2)}_{a|x}$ should be the marginals of a bipartite quantum state $\rho_{a|x}$.
\item[(ii)] We need that these bipartite quantum states satisfy $\sum_{a} p(a|x) \rho_{a|x} = \sum_{a} p(a|x) \rho_{a|x'}$ for all $x,x'$  (strongly-non-signalling  from the Alices to the Bobs).  
\item[(iii)] We need that the assemblage $\{(p(a|x),\rho_{a|x})\}$ in a scenario with a single Bob  (i.e., all Bobs are thought of as a single party)  admits of a quantum realisation\footnote{This statement is trivially satisfied in the scenario where we have only one Alice -- as the one made explicit here -- due to the GHJW theorem. However, this condition is not trivial for scenarios where we have more than one Alice.}. 
\end{compactitem}
If given an assemblage $\As_\T$ there does not exist a collection of bipartite states $\{\rho_{a|x}\}_{a,x}$ satisfying conditions (i) to (iii), then $\As_\T$ is postquantum. This can be formalised in the following theorem.

\

\begin{thm}\label{thm:hasmarginal}[\textbf{Postquantumness certification}] \quad \\
Consider a steering scenario $\T$, and specify an assemblage $\As_\T$ by the parameters
\begin{align}
\As_\T = \left\{ \left( \sigma^{(1)}_{\mathbf{a}|\mathbf{x}} ,\ldots, \sigma^{(N)}_{\mathbf{a}|\mathbf{x}}\right)  \right\}_{\mathbf{a},\mathbf{x}}\,,
\end{align}
where recall that $p(\mathbf{a}|\mathbf{x}) = \tr{\sigma^{(k)}_{\mathbf{a}|\mathbf{x}}}{} \quad \forall \, k, \mathbf{a}, \mathbf{x}$.

Then, $\As_\T$ admits of a quantum realisation if and only if the following conditions hold: 
\begin{compactenum}
\item There exists an assemblage of $N$-partite  (possibly subnormalised)  states $\sigma_{\mathbf{a}|\mathbf{x}}$ in $\cH_1 \otimes \ldots \otimes \cH_N$ (where $\cH_k$ denotes the Hilbert space of $\sigma^{(k)}_{\mathbf{a}|\mathbf{x}}$), such that
\begin{align}\label{eq:hasmarginal2}
\tr{\sigma_{\mathbf{a}|\mathbf{x}}}{\otimes_{j \neq k} \cH_j} = \sigma^{(k)}_{\mathbf{a}|\mathbf{x}} \quad \forall \, k, \mathbf{a}, \mathbf{x}\,, 
\end{align}
\item The $n+1$-partite assemblage $\As = \{\sigma_{\mathbf{a}|\mathbf{x}}\}$ where the Bobs are thought of as a single party admits of a quantum realisation.  
\end{compactenum}
Moreover, if the scenario $\T$ is such that $n=1$ (i.e., there's only one Alice), then condition 1 is not only necessary but also sufficient for $\As_\T$ to be quantumly realisable. 
\end{thm}
 
\begin{proof}
 
The `if' part follows directly by construction: condition 2 provides a quantum model, which condition 1 guarantees recovers the correct assemblage elements. 

On the other hand, if $\As_\T$ is quantumly-realisable, then its quantum model can be leveraged to define 
\begin{align*}
\sigma_{\mathbf{a}|\mathbf{x}} &= \Tr_{\otimes_{j=1:n}\cH_{A_j}} \left\{\left(\otimes_{j=1:n}M^{(j)}_{a_j|x_j} \right) \otimes \id_{d_1 \ldots d_N} \, \rho \right\}\,,
\end{align*}
which satisfies both conditions 1 and 2. 

Finally, notice that when there is only one Alice, then the GHJW theorem makes condition 1 imply condition 2, which makes condition 1 sufficient for $\As_\T$ to be quantumly realisable.
\blk
\end{proof}

\

Notice that Thm.~\ref{thm:hasmarginal} implies that, for the case of a single Alice, certifying post-quantumness of an assemblage amounts to a single instance of a semidefinite program (SDP). We will leverage this in Section \ref{se:ex} to certify post-quantum assemblages. In general, for the case where we have more than one Alice, a certification protocol can be implemented  as we discuss in Sec.~\ref{se:steNPA}.

\subsection{Certifying postquantumness from shared correlations: subtleties}\label{se:pqnwpqs}

A natural question is whether one can allow the Bobs to make measurements in their shares of a system and then assess postquantumness using this extra information. For this, the $n+N$ parties would have to do various rounds of the experiment and then come together to compute the statistics $p(\mathbf{a}\mathbf{b}|\mathbf{x}\mathbf{y})$, where $y_k$ denotes the measurement choice of the $k$-th Bob, $b_k$ denotes the outcome he obtains, and $\mathbf{b}=b_1 \ldots b_N$, $\mathbf{y}=y_1 \ldots y_N$.

If one pursues this method, then one will assess the non-quantality (i.e.,non quantum-realisability) of the assemblage $\As_\T$  using not only the information contained in $\As_\T$ but also the new information brought in by $p(\mathbf{a}\mathbf{b}|\mathbf{x}\mathbf{y})$ (which cannot be inferred from $\As_\T$).  As we see in the next example, this extra information enables the correlation-based method to give a different assessment to the one that can be concluded from Thm.~\ref{thm:hasmarginal}. That is, sometimes $\As_T$ can admit of a quantum realisation, but then if you bring onto the table more information in the form of `correlations among all the parties' then the joint statistics can no longer be explained quantumly. This might sound unnatural\footnote{The reader who is more familiar with Bell experiments may be reassured by the following. Consider a Bell scenario with two parties (two Alices). If one has only access to the two marginals $\{p(a_1|x_1)\}$ and $\{p(a_2|x_2)\}$, one can always find a local (classical) distribution $\{p(a_1a_2|x_1x_2)\}$ compatible with them. However, if one has access to the full correlation $\{p(a_1a_2|x_1x_2)\}$ produced in the experiment, one may violate a Bell inequality.   }, but it actually is a feature of the theory-independent operational approach we pursue in this manuscript. Indeed a similar situation is encountered when exploring the classical vs.~quantum boundary, as we discuss in Ap.~\ref{app:compare}. 

\

\begin{ex}\label{ex:dif}\textbf{Quantum assemblages with postquantum Bell nonlocality.--} \quad \\
Consider the scenario $\T = (n=1,|x|=2,|a|=2,N=2,d_1=2,d_2=2)$. Consider the following qubit assemblage $\As_T$ generated in the toy-theory Witworld \cite{cavalcanti2022post}, by Alice performing the two dichotomic measurements $\{\ket{0}\bra{0},\ket{1}\bra{1}\}$  and $\{\ket{+}\bra{+},\ket{-}\bra{-}\}$ on her share of a tripartite system prepared in the state $O_W$ given by:
\begin{align}
O_W = \frac{1}{4 - 8\epsilon} (\Pi_{\text{UPB}} - \epsilon \, \id)\,.
\end{align}
$\Pi_{\text{UPB}}$ is the projector onto the three-qubit subspace spanned by  $\{\ket{000},\ket{1-+},\ket{+1-},\ket{-+1}\}$. In addition, $\epsilon = \min_{\ket{\alpha \beta \gamma}} \bra{\alpha \beta \gamma}\Pi_{\mathrm{UPB}} \ket{\alpha \beta \gamma} \sim 0.0814$, with $\ket{\alpha}$, $\ket{\beta}$, and $\ket{\gamma}$ arbitrary single qubit states and $\ket{\alpha \beta \gamma}:=\ket{\alpha}\otimes\ket{\beta}\otimes\ket{\gamma}$. 

\bigskip

The assemblage $\As_\T$ constructed in this way is  strongly-non-signalling   and   has the following elements\footnote{The assemblage is presented in the Matlab workspace ABB\_PQNL.mat in the repository of Ref.~\cite{MW}. } 
\begin{align} \nonumber
\sigma^{(1)}_{0|0}=
\begin{bmatrix}
\tfrac{2924286153215233}{9007199254740992} &
\tfrac{2689945358119939}{36028797018963968} \\
& \\
\tfrac{2689945358119939}{36028797018963968} & 
\tfrac{6317253896621053}{36028797018963968}
\end{bmatrix}
\,,&\quad 
\sigma^{(1)}_{1|0} =
\begin{bmatrix}
\tfrac{6317253896621053}{36028797018963968} &
-\tfrac{2689945358119939}{36028797018963968} \\
& \\
-\tfrac{2689945358119939}{36028797018963968} & 
\tfrac{2924286153215233}{9007199254740992}
\end{bmatrix}
\,,\\   \nonumber
\quad \\   \nonumber
\sigma^{(1)}_{0|1} =
\begin{bmatrix}
\tfrac{1579313474155263}{9007199254740992} &
-\tfrac{2689945358119939}{36028797018963968} \\
& \\
-\tfrac{2689945358119939}{36028797018963968} &   
\tfrac{2924286153215233}{9007199254740992}
\end{bmatrix}
\,,&\quad
\sigma^{(1)}_{1|1} =
\begin{bmatrix}
\tfrac{2924286153215233}{9007199254740992} &
\tfrac{2689945358119939}{36028797018963968} \\
& \\
\tfrac{2689945358119939}{36028797018963968} &
\tfrac{6317253896621053}{36028797018963968}
\end{bmatrix}
\,,\\  \nonumber
\quad \\   \nonumber
\sigma^{(2)}_{0|0} =
\begin{bmatrix}
\tfrac{2924286153215233}{9007199254740992} & 
-\tfrac{2689945358119939}{36028797018963968} \\
& \\
-\tfrac{2689945358119939}{36028797018963968} &
\tfrac{6317253896621053}{36028797018963968}
\end{bmatrix}
\,,&\quad
\sigma^{(2)}_{1|0} =
\begin{bmatrix}
\tfrac{6317253896621053}{36028797018963968} &
\tfrac{2689945358119939}{36028797018963968} \\
& \\
\tfrac{2689945358119939}{36028797018963968} &  
\tfrac{2924286153215233}{9007199254740992}
\end{bmatrix}
\,,\\ \nonumber
\quad \\   \nonumber
\sigma^{(2)}_{0|1} =
\begin{bmatrix}
\tfrac{2924286153215233}{9007199254740992} & 
-\tfrac{2689945358119939}{36028797018963968} \\
& \\
-\tfrac{2689945358119939}{36028797018963968} &
\tfrac{6317253896621053}{36028797018963968}
\end{bmatrix}
\,,&\quad
\sigma^{(2)}_{1|1} =
\begin{bmatrix}
\tfrac{1579313474155263}{9007199254740992} & 
\tfrac{2689945358119939}{36028797018963968} \\
& \\
\tfrac{2689945358119939}{36028797018963968} &
\tfrac{2924286153215233}{9007199254740992}
\end{bmatrix}\,.
\end{align}
If the two Bobs were each to measure his qubit system with the two dichotomic measurements $\{\ket{0}\bra{0},\ket{1}\bra{1}\}$  and $\{\ket{+}\bra{+},\ket{-}\bra{-}\}$, then the correlations $\{p(ab_1b_2|xy_1y_2)\}$ the parties obtain would  not admit of a quantum realisation \cite{acin2010unified}. 

\bigskip

If you now run the SDP from Thm.~\ref{thm:hasmarginal}, one however finds that the assemblage $\As_\T$ does admit of a quantum realisation. The two-qubit assemblage that recovers Bobs' states as marginals may be found in the corresponding Matlab workspace.  
\end{ex}

We hence see that the assessment of the  quantum-realisability (a.k.a.~\textit{quantality})  of $\As_\T$ based solely on the assemblage elements cannot attest to the quantality of the correlations that the parties would obtain if they processed the assemblage further as in a Bell experiment.

Example \ref{ex:dif} highlights a unique feature of this theory-independent operational approach to studying steering with many characterised parties. An interesting question is how to complement the information presented in $\As_\T$ so that an assemblage deemed quantumly-realisable cannot generate postquantum correlations. In a way, this information needs to encode the way in which the Bobs are correlated, which is information not present at the local level. In addition, this might require the specification of a global $N$-partite object that somehow supersedes the assemblage $\As_\T$ itself, but how to define this in a theory-independent way is an open question. If, however, one has a specific theory in mind, then one can take this `superseding object' as the starting point, and make it be directly the state of an $N$-partite system in that theory (which should locally look quantum).

\subsection{Postquantum steering without postquantum nonlocality}\label{se:pqswpqn}

Another natural question is whether there exists an assemblage $\As_\T$ that is not quantumly-realisable (as per Thm.~\ref{thm:hasmarginal}), but whose correlations $p(\mathbf{a}\mathbf{b}|\mathbf{x}\mathbf{y})$  may  admit of a quantum realisation. Such a question has been pursued in the literature for other types of steering scenarios \cite{sainz2015postquantum,sainz2018formalism,sainz2020bipartite},  which feature only one characterised party. There, one can indeed show that some postquantum assemblages will always only yield quantumly-realisable correlations when Bob is allowed to measure. 
When one has multiple characterised parties, though, such a strong statement cannot be made without making some extra assumption.  Hence, in this paper we show that for the EPR steering scenario with many Bobs, there exist postquantum assemblages with the following properties: 
there exists a toy theory where such postquantum assemblage can arise, and may do so in a way that the correlations $p(\mathbf{a}\mathbf{b}|\mathbf{x}\mathbf{y})$ observed (were the Bobs to measure) would admit of a quantum explanation (see Sec.~\ref{se:ex}, Example \ref{ex:ABBPTP}).

In this section, we discuss how to adapt the method of Ref.~\cite{sainz2018formalism} for showing that postquantum steering in scenarios with many Bobs does not  necessarily imply   postquantum nonlocality. In Sec.~\ref{se:ex} we leverage this method to find examples. 
For simplicity in the discussion we consider a scenario $\T = (n=1,|x|,|a|,N=2,d_1,d_2)$ -- i.e., one Alice and two Bobs. The techniques, however, generalise straightforwardly to scenarios with a larger number of Alices and/or Bobs.

\begin{prop}\label{prop:wit}\textbf{Assemblages  compatible  with  quantum Bell-type   nonlocality.--}\\
Consider a tripartite quantum system $\cH = \cH_A\otimes \cH_1 \otimes \cH_2$ shared by Alice, the first Bob, and the second Bob, respectively. Let $\Lambda^{(k)}[\cdot]$ be a Positive and Trace-Preserving (PTP) map on $\cH_k$, which may not be completely positive, for each $k=1,2$. Let $\{\{M_{a|x}\}_a\}_x$ be a collection of (complete) measurements in $\cH_A$, and $\rho$ a (normalised) density matrix in $\cH$. Using the toy theory Witworld\cite{cavalcanti2022post}, construct the  strongly-non-signalling    assemblage $\As = \{\left( \sigma^{1}_{a|x}, \sigma^{2}_{a|x}\right)\}_{a,x}$ with elements:
\begin{align}
\sigma^{k}_{a|x} &= \tr{ \left(M_{a|x} \otimes \id_1 \otimes \id_2 \right) (\id_A \otimes \Lambda^{(1)} \otimes \Lambda^{(2)})[\rho] }{\cH_A,\cH_{\neg k}}\,,
\end{align}
where $\id_k$ is the identity in $\cH_k$. 

Then, the correlations $p(a\mathbf{b}|x\mathbf{y})$  generated when the Bobs perform local measurements on their systems always admit of a quantum realisation. 
\end{prop}
 Notice that here we do not only have a specification of the assemblage $\As$, but also a realisation of this assemblage within the toy theory Witworld, which we leverage to make inferences on what the correlations among the Bobs are. It could happen that within a different theory the same assemblage can be prepared but in a way that the correlations $p(a\mathbf{b}|x\mathbf{y})$ do not admit of a quantum realisation.  
\begin{proof}
Let Alice and the two Bobs take the assemblage $\As$ to perform a Bell-type experiment. Denote by $\mathbf{M}^{(1)} = \{M^{(1)}_{b_1|y_1}\}_{b_1, x_1}$ and $\mathbf{M}^{(2)} = \{M^{(2)}_{b_2|y_2}\}_{b_2, x_2}$ the operators corresponding to the measurements performed by the first and second Bob, respectively. 

Notice that, locally, applying a PTP map $\Lambda^{(k)}[\cdot]$ followed by a measurement in $\mathbf{M}^{(k)}$ yields the same statistics as directly applying the well-defined measurement from $\mathbf{\tilde{M}}^{(k)} = \{\tilde{M}^{(k)}_{b_k|y_k}\}_{b_k, x_k}$ with elements 
$\tilde{M}^{(k)}_{b_k|y_k} = \left( \Lambda^{(k)}\right)^\dagger [{M}^{(k)}_{b_k|y_k}]$\,.
Hence, the statistic observed from performing $\mathbf{M}^{(1)}$ and $\mathbf{M}^{(2)} $ on $\As$ is the same as the statistics observed from performing $\mathbf{\tilde{M}}^{(1)}$ and $\mathbf{\tilde{M}}^{(2)} $ on the \textit{quantum assemblage} $\tilde{\As}$ with elements 
\begin{align*}
\tilde{\sigma}^{k}_{a|x} &= \tr{ \left(M_{a|x} \otimes \id_1 \otimes \id_2 \right) \rho }{\cH_A,\cH_{\neg k}}\,.
\end{align*}
Hence, within Witworld the correlations $p(a\mathbf{b}|x\mathbf{y})$ admit of a quantum realisation of the form
$p(a\mathbf{b}|x\mathbf{y}) = \tr{ \left(M_{a|x} \otimes \tilde{M}^{(1)}_{b_1|y_1} \otimes \tilde{M}^{(2)}_{b_2|y_2} \right) \rho }{}$. 
\end{proof}

This mathematical property of locally-applied PTP maps was leveraged in Ref.~\cite{sainz2018formalism} to find examples of multipartite assemblages (with only one Bob) whose postquantumness cannot be detected by looking at the correlations that may arise when Bob further measures his system.  In those cases, however, the claim holds in full generality, regardless of which theory underpins the experiment.  Here, a  related   opportunity arises: even though $p(a\mathbf{b}|x\mathbf{y})$ may admit of a quantum realisation  in some toy theory,  the fact that the operator $\tilde{\rho} = (\id_A \otimes \Lambda^{(1)} \otimes \Lambda^{(2)})[\rho]$ may be negative (i.e., not a valid quantum state) makes it possible that the corresponding assemblage $\As$ defies a quantum explanation. In Sec.~\ref{se:ex} we present examples of assemblages $\As$ constructed as in Prop.~\ref{prop:wit} whose postquantumness is certified via Thm.~\ref{thm:hasmarginal}.

\section{Relevant examples of postquantum steering}\label{se:ex}

In this section we present a variety of examples of postquantum steering. The numerics here and in the next section were carried out in Matlab \cite{MATLAB}, using the software CVX \cite{grant2013cvx}, the solver SDPT3 \cite{tutuncu1999sdpt3} and the toolbox QETLAB \cite{qetlab}; see Ap.~\ref{ap:A} and the repository in Ref.~\cite{MW}.

First let us present an example of a postquantum assemblage in steering scenario with only one Alice. This leverages the approach to postquantum steering certification presented in Sec.~\ref{se:ntc}, which makes this a non-trivial example of post-quantum steering -- that is, one that doesn't leverage Bell nonlocality  among the Alices  to make the assessment. 

\begin{ex}\label{ex:pqs1} \textbf{Postquantum steering example based on global states.--}\\
Consider the scenario $\T = \left(n=1,|\X|=2,|\A|=2,N=2,d_1=2,d_2=2 \right)$, and take  $\X_1 = \X_2 = \A^j_k = \{0,1\}$. Define the assemblage\footnote{The assemblage is presented in the Matlab workspace ABB\_1.mat.} $\As_\T = \{ \sigma^{(1)}_{a|x}, \sigma^{(2)}_{a|x})  \}_{a|x}$ as

\begin{align}
\sigma^{(1)}_{0|0} &= \begin{bmatrix}
\tfrac{86505229615495}{281474976710656} &
\tfrac{8644919822415911}{36028797018963968} + \tfrac{463283299018779}{18014398509481984}i \\
& \\
\tfrac{8644919822415911}{36028797018963968} - \tfrac{463283299018779}{18014398509481984}i &                                     
\tfrac{7512381172927199}{36028797018963968}
\end{bmatrix} \,,\\  \nonumber
\quad \\ \nonumber
\sigma^{(1)}_{1|0} &= \begin{bmatrix}
\tfrac{37723717191037}{140737488355328} &
\tfrac{5332123403232109}{2305843009213693952} + \tfrac{7972126598771139}{36028797018963968}i \\
& \\
\tfrac{5332123403232109}{2305843009213693952} - \tfrac{7972126598771139}{36028797018963968}i &                                          
\tfrac{243327339198373}{1125899906842624}
\end{bmatrix} \,, 
\end{align}
\begin{align} \nonumber
\sigma^{(1)}_{0|1} &= \begin{bmatrix}
\tfrac{8506857505773515}{36028797018963968}  &
\tfrac{3381940540593419}{36028797018963968} + \tfrac{1630332021704117}{72057594037927936}i \\
& \\
\tfrac{3381940540593419}{36028797018963968} - \tfrac{1630332021704117}{72057594037927936}i  &                                      
\tfrac{3751794212825959}{18014398509481984}
\end{bmatrix} \,,\\  \nonumber
\quad \\ \nonumber
\sigma^{(1)}_{1|1} &= \begin{bmatrix}
\tfrac{1527885435739415}{4503599627370496} &
\tfrac{2673146854998997}{18014398509481984} + \tfrac{4041763592978319}{18014398509481984} i \\
& \\
\tfrac{2673146854998997}{18014398509481984} - \tfrac{4041763592978319}{18014398509481984} i &                                        
\tfrac{487204225101451}{2251799813685248}
\end{bmatrix} \,,\\  \nonumber
\quad \\ 
\sigma^{(2)}_{0|0} &= \begin{bmatrix}
\tfrac{2038719144082355}{4503599627370496} &
\tfrac{7589461495857015}{72057594037927936} + \tfrac{6552116316195633}{1152921504606846976}i \\
& \\
\tfrac{7589461495857015}{72057594037927936} - \tfrac{6552116316195633}{1152921504606846976}i &
\tfrac{284412176381465}{4503599627370496}
\end{bmatrix} \,,\\  \nonumber
\quad \\ \nonumber
\sigma^{(2)}_{1|0} &= \begin{bmatrix}
\tfrac{1209315864020551}{4503599627370496} &
\tfrac{6624831370177975}{72057594037927936} + \tfrac{992983320980893}{4503599627370496}i \\
& \\
\tfrac{6624831370177975}{72057594037927936} - \tfrac{992983320980893}{4503599627370496}i &
\tfrac{971152442886125}{4503599627370496}
\end{bmatrix} \,, \\ \nonumber
\quad \\ \nonumber
\sigma^{(2)}_{0|1} &= \begin{bmatrix}
\tfrac{687889890777155}{2251799813685248} &
\tfrac{3855767902614183}{288230376151711744} + \tfrac{1339362793814727}{9007199254740992}i \\
& \\
\tfrac{3855767902614183}{288230376151711744} - \tfrac{1339362793814727}{9007199254740992}i &
\tfrac{5004207678990951}{36028797018963968} 
\end{bmatrix} \,,\\  \nonumber
\quad \\ \nonumber
\sigma^{(2)}_{1|1} &= \begin{bmatrix}
\tfrac{468063806637149}{1125899906842624} &
\tfrac{6625175445190723}{36028797018963968} + \tfrac{5582338054938701}{72057594037927936}i \\
& \\
\tfrac{6625175445190723}{36028797018963968} - \tfrac{5582338054938701}{72057594037927936}i &                                      
\tfrac{630038659393721}{4503599627370496}
\end{bmatrix} \,.
\end{align}

\bigskip 

This  weakly-non-signalling  assemblage $\As_\T$ is postquantum, as certified by the procedure of Thm.~\ref{thm:hasmarginal}: there is no two-qubit ensemble of (possibly subnormalised) states $\{\sigma_{a|x}\}_{a,x}$ such that constraints (1) and (2) are satisfied. 

\end{ex}

Another similar example\footnote{The assemblage is presented in the Matlab workspace ABB\_2.mat.} will be also discussed in the next section.

\bigskip

Next we present an example\footnote{The assemblage is presented in the Matlab workspace ABB\_PTP\_1.mat.} of an assemblage that is  both  strongly-non-signalling  and postquantum, but  where the correlations it produces in a traditional Bell scenario could accept a quantum explanation,  as per discussed in Sec.~\ref{se:pqswpqn}.

\begin{ex}\label{ex:ABBPTP} \textbf{Postquantum steering example with quantum nonlocality.--}\\
Consider a steering scenario $\T = \left(n=1,|\X|=2,|\A|=2,N=2,d_1=2,d_2=4 \right)$. Consider the assemblage $\As_\T = \left\{ \left( \sigma^{(1)}_{\mathbf{a}|\mathbf{x}} , \sigma^{(2)}_{\mathbf{a}|\mathbf{x}}\right)  \right\}_{\mathbf{a},\mathbf{x}}$ mathematically specified\footnote{This is not a protocol to physically implement the assembalge within quantum theory, since the map $\Lambda^{(2)}$ is not a completely-positive map. This is just a convenient mathematical description of $\As_\T$.  } as in Prop.~\ref{prop:wit} with the following: 
\begin{compactitem}
\item the Hilbert spaces $\cH_A$ and $\cH_1$ have dimension $2$, and $\cH_2$ has dimension 4. 
\item The state $\rho$ in $\cH_A \otimes \cH_1 \otimes \cH_2$ is defined as $\rho = \ket{\psi}\bra{\psi}$, with $\ket{\psi}$ given by
\begin{align*}
\ket{\psi} = 
[\tfrac{90}{20389}  +  \tfrac{374}{3209}i,\, 
    -\tfrac{228}{1439}   -  \tfrac{116}{987}i  ,\,
     \tfrac{933}{4015}   +  \tfrac{241}{10931}i ,\,
    -\tfrac{659}{1299}   -  \tfrac{983}{8054}i ,    \\
     \tfrac{190}{9687}   -  \tfrac{289}{1353}i ,\,
    -\tfrac{352}{8837}   +   \tfrac{62}{621}i  ,\,
    -\tfrac{356}{2211}   -  \tfrac{269}{1254}i ,\,
     \tfrac{171}{3658}   +  \tfrac{358}{2201}i ,     \\
    -\tfrac{280}{1671}   -  \tfrac{108}{2369}i ,\,
    -\tfrac{221}{2374}   +  \tfrac{121}{1803}i ,\,
     \tfrac{457}{2027}   -  \tfrac{415}{1774}i ,\,
    -\tfrac{392}{1353}   -  \tfrac{338}{2799}i ,    \\
     \tfrac{710}{7719}   -  \tfrac{226}{1985}i ,\,
     \tfrac{214}{2159}   -  \tfrac{211}{1007}i ,\,
    -\tfrac{297}{3218}   +  \tfrac{259}{774}i  ,\,
     \tfrac{417}{3082}   -   \tfrac{61}{1741}i ]^\dagger\,.
\end{align*}
\item The PTP map $\Lambda^{(1)}[\cdot]$ is taken to be just the identity. The PTP map $\Lambda^{(2)}[\cdot]$ is given by
\begin{align}
\Lambda^{(2)}[\rho] = \tfrac{1}{2} \left\{ \tr{\rho}{} \id - \rho - U \rho^{\mathrm{T}} U^\dagger \right\}\,,
\end{align}
where ${}^{\mathrm{T}}$ denotes transposition, and $U = X \otimes Y$ is an antisymmetric unitary with $X$ and $Z$ the Pauli operators. 
\item Alice's measurement operators $\mathbf{M}$ are given by $\{M_{0|0} = \ket{+}\bra{+}, M_{1|0} = \ket{-}\bra{-}\}$ and $\{M_{0|1} = \ket{0}\bra{0}, M_{1|1}=\ket{1}\bra{1}\}$. 
\end{compactitem}
One can check that this assemblage is well defined in $\T$, and by construction  it is compatible with quantumly-realisable correlations when underpinned by the toy theory Witworld  (see Sec.~\ref{se:pqswpqn}). This assemblage does not admit of a quantum realisation, as certified by Thm.~\ref{thm:hasmarginal}, since there is no ensemble of (possibly subnormalised) states $\{\sigma_{a|x}\}_{a,x}$ such that constraints (1) and (2) are satisfied. Finally, the assemblage  is strongly-non-signalling  since it may be realised as a common-cause process in Witworld.  
\end{ex}

Example \ref{ex:ABBPTP} is crucial for two reasons. On the one hand, it shows that postquantum steering in this type of scenarios is a phenomenon independent of postquantum nonlocality, which makes the phenomenon interesting in its own right. On the other hand, the assemblage in this example can arise within the toy theory Witworld, which makes postquantum steering not only a mathematical curiosity but a phenomenon featured by compositional theories beyond quantum theory, hence making its study fundamentally relevant. Another similar example\footnote{The assemblage is presented in the Matlab workspace ABB\_PTP\_2.mat.} will be also presented in the next section. 

\bigskip

Finally, one could consider the steering scenario with two Alices and two Bobs: $\T = \left(n=2,|\X|=2,|\A|=2,N=2,d_1=2,d_2=4 \right)$. 
We present now examples\footnote{The assemblages are presented in the Matlab workspaces AABB\_PTP\_1.mat and AABB\_PTP\_2.mat.} of an assemblage in such scenario that is  both  strongly-non-signalling  and  postquantum, but  whose correlations in a traditional Bell scenario admit of a quantum explanation when underpinned by the toy theory Witworld.

\begin{ex}\label{ex:AABBPTP} \textbf{Postquantum steering example with quantum nonlocality with two Alices.--}\\
Consider the steering scenario $\T = \left(n=2,|\X|=2,|\A|=2,N=2,d_1=2,d_2=4 \right)$, with two Alices and two Bobs. Construct an assemblage using the same procedure as in Example \ref{ex:ABBPTP}, but by sampling a four-partite quantum state of the system $\cH_A \otimes \cH_A \otimes \cH_1 \otimes \cH_2$ instead. We apply the PTP map $\Lambda^{(2)}$ on the quantum system $\cH_2$, and take the first Alice to apply the measurements described in Example \ref{ex:ABBPTP}. For the second Alice, we define her measurements as the projection onto the eigenstates of $X+Z$ for $x_2=0$, and onto the eigenstates of $X-Z$ for $x_2=1$.

The assemblages defined with this procedure are well defined, and by construction  may yield quantum correlations when underpinned by the theory Witworld  (see Sec.~\ref{se:pqswpqn}). Here we present two numerical examples\footnotemark[10] of assemblages generated by this procedure and that do not admit of a quantum realisation as certified by Thm.~\ref{thm:hasmarginal}.  Finally, the assemblages  are strongly-non-signalling   since they may be realised as common-cause processes in Witworld.   
\end{ex}

\section{Robustness of postquantum steering }

A natural question pertains to how robust a postquantum assemblage is; namely, how much noise can the assemblage tolerate while still defying a quantum explanation. In this section we discuss how to compute the action of noise on an assemblage, and see how two types of noise can affect some postquantum assemblages that we have discovered. 

\bigskip

More generally, the situation is that we have two assemblages $\As_\T^1$ and $\As_\T^2$, and we want to know how to represent the elements of the assemblage $\As_\T^r$ that results from mixing them with weights $r$ and $1-r$ respectively. For this it is convenient to represent the assemblages as per Eq.~\eqref{eq:asssub}. With this perspective one can hence write
\begin{align}
\left(\sigma^k_{\mathbf{a}|\mathbf{x}}\right)^r = r \, \left(\sigma^k_{\mathbf{a}|\mathbf{x}}\right)^1 + (1-r) \, \left(\sigma^k_{\mathbf{a}|\mathbf{x}}\right)^2 \,,
\end{align} 
for all characterised parties $k=1,\ldots,N$. Intuition behind this can be drawn by thinking of a quantum example where one first computes the mixture of the global states $\sigma^{1\ldots N}_{\mathbf{a}|\mathbf{x}}$ and then takes the marginals for each Bob. 

\bigskip

Here we will explore two different types of noise: white noise and `entangled' noise. White noise is generated by the Alices flipping each an unbiased coin in her lab, while the $n+N$ parties share a maximally mixed quantum state. For a steering scenario $\T$ this is formalised as
\begin{align}
\As_T^\mathrm{w} = \left\{ \left( p(\mathbf{a}|\mathbf{x}) \, \frac{\id_{d_1}}{d_1} \, \ldots, p(\mathbf{a}|\mathbf{x}) \, \frac{\id_{d_N}}{d_N}     \right) \right\} \,,
\end{align}
where\footnote{Here for simplicity in the presentation we have considered the case where for each Alice, all her measurements have the same number of outcomes (which can differ from Alice to Alice). A similar object can be formalised in full generality using more contrived notation.}  $p(\mathbf{a}|\mathbf{x}) = \prod_{k=1}^n \frac{1}{|\A^k| }$. 

For the case of `entangled' noise, we mean a quantum assemblage that is generated by the Alices performing local measurements on a $n+N$-partite system prepared in an entangled state. Different choices of states give rise to different types of noise. For examples where both Bobs have qubit systems, and there is only one Alice, it is natural to consider three-qubit systems prepared either in the GHZ state $\ket{\mathrm{GHZ}}=\frac{\ket{000}+\ket{111}}{\sqrt{2}}$ or in the W state $\ket{\mathrm{W}} = \frac{\ket{001}+\ket{010}+\ket{100}}{\sqrt{3}}$. We will take Alice to perform two dichotomic measurements in her qubit, one along the $Z$ basis ($x=0$) and one along the $X$ basis ($x=1$). The quantum assemblages that correspond to these two cases are given by: 
\begin{align}
\As_\T^{\mathrm{GHZ}} \,: \quad 
 \sigma^{(k)}_{0|0} = 
\begin{bmatrix}
    \tfrac{1}{2}    &     0 \\
         0    &     0
\end{bmatrix} 
\,,\quad 
\sigma^{(k)}_{1|0} = 
\begin{bmatrix}
         0   &      0 \\
         0   &  \tfrac{1}{2}
\end{bmatrix} 
\,,\quad
 \sigma^{(k)}_{0|1} = 
\begin{bmatrix}
    \tfrac{1}{4}  &       0 \\
         0  &  \tfrac{1}{4}
\end{bmatrix} 
\,,\quad 
\sigma^{(k)}_{1|1} = 
\begin{bmatrix}
    \tfrac{1}{4}   &      0 \\
         0   &  \tfrac{1}{4}
\end{bmatrix} \,, \quad k=1,2\,,
\end{align}
and
\begin{align}
\As_\T^{\mathrm{W}} \,: \quad 
 \sigma^{(k)}_{0|0} = 
\begin{bmatrix}
   \tfrac{1}{3} & 0 \\
   0 & \tfrac{1}{3}
\end{bmatrix} 
\,,\quad 
\sigma^{(k)}_{1|0} = 
\begin{bmatrix}
  \tfrac{1}{3} & 0 \\
  0 & 0 
\end{bmatrix} 
\,,\quad
 \sigma^{(k)}_{0|1} = 
\begin{bmatrix}
   \tfrac{1}{3} & \tfrac{1}{6} \\
   \tfrac{1}{6} & \tfrac{1}{6}
\end{bmatrix} 
\,,\quad 
\sigma^{(k)}_{1|1} = 
\begin{bmatrix}
   \tfrac{1}{3} & -\tfrac{1}{6} \\
   -\tfrac{1}{6} & \tfrac{1}{6}
\end{bmatrix} \,, \quad k=1,2\,.
\end{align}

\bigskip

Given an assemblage $\As_\T$, its robustness to noise can then be computed as the minimum amount $r$ of noise that needs to be added such that the new assemblage admits of a quantum realisation. Mathematically, this reads:
\begin{align}
\text{min} &\quad r \\ \nonumber
\text{st} &\quad r \, \left(\sigma^k_{\mathbf{a}|\mathbf{x}}\right)^\text{noise} + (1-r) \, \left(\sigma^k_{\mathbf{a}|\mathbf{x}}\right) \quad \forall\, k,\mathbf{a},\mathbf{x}\,,\quad \text{and}\\
&\quad 
\left\{ \left( \left(\sigma^{(1)}_{\mathbf{a}|\mathbf{x}}\right)^\text{noise} ,\ldots, \left(\sigma^{(N)}_{\mathbf{a}|\mathbf{x}}\right)^\text{noise} \right)  \right\}_{\mathbf{a},\mathbf{x}}\quad \text{is quantumly realisable.}
\end{align}
We will now compute the robustness to noise for the examples mentioned in the previous section, and for different types of noise. 

\bigskip

Let us begin with the scenario $\T = \left(n=1,|\X|=2,|\A|=2,N=2,d_1=2,d_2=2 \right)$. 

First take the assemblage given by Example \ref{ex:pqs1}. Here, the robustness to white noise is $r^\mathrm{w} \simeq  0.0147$, the robustness to GHZ noise is $r^\mathrm{GHZ} \simeq  0.0167$, and the robustness to W noise is $r^\mathrm{W} \simeq  0.0139$. These robustness feature $r^\mathrm{W} < r^\mathrm{w} < r^\mathrm{GHZ}$. We see then that we need more than 1\% noise to have the postquantum assemblage admit of a quantum realisation. 

Now consider the example from the Matlab workspace \text{ABB\_2.mat}. This assemblage features a robustness to white noise of $r^\mathrm{w} \simeq  0.0014$, a robustness to GHZ noise of $r^\mathrm{GHZ} \simeq  0.0012$, and a robustness to W noise of $r^\mathrm{W} \simeq  0.0020$. These robustness feature $r^\mathrm{GHZ} < r^\mathrm{w} < r^\mathrm{W}$, which is qualitatively different from the relation in the previous example. Here, however, than much less noise is already sufficient to break the postquantumness of the assemblage. 

 \bigskip

Let us continue with the scenario $\T = \left(n=1,|\X|=2,|\A|=2,N=2,d_1=2,d_2=2 \right)$, where given the dimension of the Hilbert space of the second Bob we will focus on robustness to white noise only. The assemblage in Example \ref{ex:ABBPTP} yields $r^\mathrm{w} \simeq  0.0017$. The assemblage presented in  the Matlab workspace \text{ABB\_PTP\_2.mat}, which was obtained as in Example \ref{ex:ABBPTP} but for a different state $\ket{\psi}$, yields  $r^\mathrm{w} \simeq  0.0014$. We see that in these examples, then, a little noise is already sufficient to break the postquantumness of the assemblage. 

 \bigskip
  
Finally, let us explore the robustness to noise of the postquantum assembalges in scenarios with two Alices and two Bobs described in Example \ref{ex:AABBPTP}. The assemblage from the Matlab workspace AABB\_PTP\_1.mat has a robustness to white noise $r^\mathrm{w} \simeq  0.0085$, whereas the assemblage from the Matlab workspace AABB\_PTP\_2.mat yields $r^\mathrm{w} \simeq  0.00604$. We see that these instances of postquantumness can resist only a small white noise, of the order of $1\%$. 

\section{Outer approximation of the set of quantumly-realisable assemblages}\label{se:steNPA}

In this section we discuss the natural question of how to approximate `from the outside' the set of assemblages that admit of a quantum realisation. This question has been explored for correlations in Bell scenarios \cite{navascues2007bounding,navascues2008convergent} and for assemblages in steering scenarios with only one characterised party \cite{johnston2016extended,hoban2025hierarchy}. 

By leveraging Thm.~\ref{thm:hasmarginal}, we can propose a hierarchy of sets of assemblages, each of which defined as a semidefinite program.  

\begin{defn}\label{defn:qn}[\textbf{Assemblages in the set} $\mathbf{\cQ_\T^k}$] \quad \\
Consider a steering scenario $\T$, and specify an assemblage $\As_\T$ by the parameters
\begin{align}
\As_\T = \left\{ \left( \sigma^{(1)}_{\mathbf{a}|\mathbf{x}} ,\ldots, \sigma^{(N)}_{\mathbf{a}|\mathbf{x}}\right)  \right\}_{\mathbf{a},\mathbf{x}}\,.
\end{align}

$\As_\T$ belongs to $\mathbf{\cQ_\T^k}$ if and only if the following conditions hold: 
\begin{compactenum}
\item There exists an assemblage of $N$-partite  (possibly subnormalised)  states $\sigma_{\mathbf{a}|\mathbf{x}}$ in $\cH_1 \otimes \ldots \otimes \cH_N$ (where $\cH_\ell$ denotes the Hilbert space where $\sigma^{(\ell)}_{\mathbf{a}|\mathbf{x}}$ lives), such that
\begin{align}\label{eq:hasmarginal3}
\tr{\sigma_{\mathbf{a}|\mathbf{x}}}{\otimes_{j \neq \ell} \cH_j} = \sigma^{(\ell)}_{\mathbf{a}|\mathbf{x}} \quad \forall \, \ell, \mathbf{a}, \mathbf{x}\,, 
\end{align}
\item The $n+1$-partite assemblage $\As = \{\sigma_{\mathbf{a}|\mathbf{x}}\}$ where the Bobs are thought of as a single party belongs to the $k$-th level of the Johnston-Mittal-Russo-Watrous hierarchy\footnote{Alternatively, the $k$-th level of the hierarchy of Ref.~\cite{hoban2025hierarchy} when $y$ takes values in the singleton set. } \cite{johnston2016extended}. 
\end{compactenum}
\end{defn}
By definition, then, $\mathbf{\cQ_\T^{k+1}} \subseteq \mathbf{\cQ_\T^k}$, and $\bigcap_{k=1:\infty}\mathbf{\cQ_\T^k} = \mathbf{\cQ_\T}$. 

\

Finally, notice that Def.~\ref{defn:qn} opens the door to specifying so-called `almost-quantum assemblages' \cite{hoban2018channel,rossi2022characterising} in these EPR scenarios, by focusing on the usually called `1+AB' level of the hierarchy. The formalisation of this idea is left for future work. 

\section{Conclusions}

This paper studies postquantum steering in scenarios with multiple characterised parties (Bobs). Our work brings in the missing piece of the puzzle in the operational study of postquantum steering, since all other works focus on the specific case of one single quantum Bob. 
The first and crucial step we take is to identify the mathematical object that represents an assemblage from an operational (and arguably theory-independent) viewpoint in these scenarios. What we denote as an assemblage here, hence, differs from what has been studied in the literature on quantum vs.~classical steering with multiple characterised parties. In particular, it was necessary in this work to articulate  in which relevant ways an assemblage may be non-signalling.   

Here we focus on the scope and properties of postquantum steering, the boundary between quantum and postquantum, and the relation between postquantum steering and postquantum Bell nonlocality. We find that postquantum steering is formally possible, and present several examples. Some of these examples may emerge within the foil theory Witworld \cite{cavalcanti2022post}, which renders the phenomenon fundamentally relevant and not a mere mathematical curiosity. To certify postquantum steering we designed an algorithm which, for the case of a single uncharacterised party (Alice), also serves as a certificate of quantumness. That is, when we have only one Alice, the algorithm singles out the boundary of the quantum set of assemblages. 
In addition, for the case of two or more Alices, we provide a hierarchy of semidefinite programs, where each level gives an outer-approximation to the set of quantumly-realisable assemblages.
The robustness of our examples to different types of noise was analysed; we found that these were not very large (of the order of 0.1\%-1\%) from the viewpoint of physical implementations (which are, however, not relevant for the scope of this manuscript), but nonetheless are substantial enough to show our examples are not mere numerical instabilities. 

The connection with Bell nonlocality, however, is more subtle than in the case of the previously-studied scenarios. One similar trait all these scenarios prove to share (ours included) is that there exist postquantum assembalges that  are compatible with quantumly-realisable Bell correlations,  for the Bell scenario where the Bobs perform measurements on their quantum systems as well. This makes postquantum steering in these scenarios an interesting and relevant phenomenon on its own right, which is not a mere consequence of Bell postquantumness. Steering scenarios with multiple Bobs, however, have a peculiar trait that other scenarios can't feature: here one can have quantumly-realisable assemblages that give rise to postquantum Bell correlations when the Bobs measure their quantum systems. A similar trait appears when comparing quantum and classical steering: some assemblages admit of a classical model but may give rise to correlations that violate Bell inequalities (see Ap.~\ref{app:compare}). 
 Hence, we see that our operational theory-independent approach to steering makes the phenomenon incomparable to that of Bell nonlocality.  
This fact, which is highly counter-intuitive, stems from the fact that an assemblage in our scenario has no information on how the Bobs are correlated, hence one cannot infer what will happen if they were to measure. 
 One way to make this inference is to promote the steering experiment to a Bell one, run the statistics again, and compute the correlations $p(\mathbf{a}\mathbf{b}|\mathbf{x}\mathbf{y})$ from scratch (which, if we have a quantum experiment to begin with will not yield to postquantumness). Another way is to find realisations of the assemblage within other foil theories (such as in Example \ref{ex:dif} where the assemblage could arise not only from a multipartite quantum system but also from a multipartite Witworld system prepared in an entanglement witness), and see whether the statistics $p(\mathbf{a}\mathbf{b}|\mathbf{x}\mathbf{y})$ may defy a quantum explanation. This approach, however, is no longer theory-independent, and relies on the availability of suitable alternative foil theories. Similarly, for the case of classical vs.~quantum assemblages, one can equip the scenario with a quantum theoretical description (this is what is usually done), and hence all information on how the Bobs correlate may be available (though the phenomenon is no longer studied in a theory-independent way). 

Looking forward, our paper has opened the door to the study of steering in a different type of scenarios, from an operational and theory-independent viewpoint. This is the first time (we know of) that this object has been formalised, and hence there is a variety of questions one could pursue. 
On the one hand, one could perform a deeper study of classical vs.~nonclassical steering, beyond what we discussed in this manuscript. One could also explore definitions for other types of families of assemblages, and certification algorithms. On the other hand, one could explore the resourcefulness of this type of quantum and postquantum steering for communication or information processing tasks, including the development of resource theories (equipped with conversion algorithms or resource monotones) to formally quantify steering (quantum or postquantum) as a resource.

\section*{Acknowledgments}

I'm grateful for the insightful conversations with Andreas Winter about non-signaling conditions.
I thank John H.~Selby, David Schmid, and Vinicius P.~Rossi for useful discussions at the initial stages of the project. 
I also thank John H.~Selby and Matty J.~Hoban for feedback on an earlier version of this manuscript. 
I am grateful to Adriana Micelli, Daniel Sainz, Linda Selby, and Philip Selby, who provided the babysitting help I needed to make progress on this project. I am also grateful to Luca Sainz Selby for behaving very well with the grandparents, making all this possible.
This work is carried out under IRA Programme, project no.~FENG.02.01-IP.05-0006/23, financed by the FENG program 2021-2027, Priority FENG.02, Measure FENG.02.01., with the support of the FNP.

\bibliographystyle{quantum}
\bibliography{pqsmbrefs}

\begin{thebibliography}{10}

\bibitem{schrodinger1935discussion}
Erwin Schr{\"o}dinger.
\newblock ``Discussion of probability relations between separated systems''.
\newblock
  \href{https://dx.doi.org/https://doi.org/10.1017/S0305004100013554}{Mathematical
  Proceedings of the Cambridge Philosophical Society {\bf 31},
  555--563}~(1935).

\bibitem{schrodinger36}
E.~Schrödinger.
\newblock ``Probability relations between separated systems''.
\newblock \href{https://dx.doi.org/10.1017/S0305004100019137}{Mathematical
  Proceedings of the Cambridge Philosophical Society {\bf 32},
  446–452}~(1936).

\bibitem{wiseman2007steering}
Howard~M Wiseman, Steve~James Jones, and Andrew~C Doherty.
\newblock ``Steering, entanglement, nonlocality, and the
  {E}instein-{P}odolsky-{R}osen paradox''.
\newblock
  \href{https://dx.doi.org/https://doi.org/10.1103/PhysRevLett.98.140402}{Physical
  review letters {\bf 98}, 140402}~(2007).

\bibitem{branciard2012one}
Cyril Branciard, Eric~G Cavalcanti, Stephen~P Walborn, Valerio Scarani, and
  Howard~M Wiseman.
\newblock ``One-sided device-independent quantum key distribution: {S}ecurity,
  feasibility, and the connection with steering''.
\newblock
  \href{https://dx.doi.org/https://doi.org/10.1103/PhysRevA.85.010301}{Physical
  Review A {\bf 85}, 010301}~(2012).

\bibitem{supic2016}
Ivan {\v{S}}upi{\'c} and Matty~J Hoban.
\newblock ``Self-testing through {EPR}-steering''.
\newblock
  \href{https://dx.doi.org/https://doi.org/10.1088/1367-2630/18/7/075006}{New
  Journal of Physics {\bf 18}, 075006}~(2016).

\bibitem{gheorghiu2017}
Alexandru Gheorghiu, Petros Wallden, and Elham Kashefi.
\newblock ``Rigidity of quantum steering and one-sided device-independent
  verifiable quantum computation''.
\newblock
  \href{https://dx.doi.org/https://doi.org/10.1088/1367-2630/aa5cff}{New
  Journal of Physics {\bf 19}, 023043}~(2017).

\bibitem{law2014quantum}
Yun~Zhi Law, Jean-Daniel Bancal, Valerio Scarani, et~al.
\newblock ``Quantum randomness extraction for various levels of
  characterization of the devices''.
\newblock
  \href{https://dx.doi.org/https://doi.org/10.1088/1751-8113/47/42/424028}{Journal
  of Physics A: Mathematical and Theoretical {\bf 47}, 424028}~(2014).

\bibitem{passaro2015optimal}
Elsa Passaro, Daniel Cavalcanti, Paul Skrzypczyk, and Antonio Ac{\'\i}n.
\newblock ``Optimal randomness certification in the quantum steering and
  prepare-and-measure scenarios''.
\newblock
  \href{https://dx.doi.org/https://doi.org/10.1088/1367-2630/17/11/113010}{New
  Journal of Physics {\bf 17}, 113010}~(2015).

\bibitem{quintino2014joint}
Marco~T\'ulio Quintino, Tam\'as V\'ertesi, and Nicolas Brunner.
\newblock ``Joint measurability, einstein-podolsky-rosen steering, and bell
  nonlocality''.
\newblock \href{https://dx.doi.org/10.1103/PhysRevLett.113.160402}{Phys. Rev.
  Lett. {\bf 113}, 160402}~(2014).

\bibitem{uola2014joint}
Roope Uola, Tobias Moroder, and Otfried G\"uhne.
\newblock ``Joint measurability of generalized measurements implies
  classicality''.
\newblock \href{https://dx.doi.org/10.1103/PhysRevLett.113.160403}{Phys. Rev.
  Lett. {\bf 113}, 160403}~(2014).

\bibitem{cavalcanti2016quantitative}
D.~Cavalcanti and P.~Skrzypczyk.
\newblock ``Quantitative relations between measurement incompatibility, quantum
  steering, and nonlocality''.
\newblock \href{https://dx.doi.org/10.1103/PhysRevA.93.052112}{Phys. Rev. A
  {\bf 93}, 052112}~(2016).

\bibitem{popescu1994quantum}
Sandu Popescu and Daniel Rohrlich.
\newblock ``Quantum nonlocality as an axiom''.
\newblock
  \href{https://dx.doi.org/https://doi.org/10.1007/BF02058098}{Foundations of
  Physics {\bf 24}, 379--385}~(1994).

\bibitem{brunner2014bell}
Nicolas Brunner, Daniel Cavalcanti, Stefano Pironio, Valerio Scarani, and
  Stephanie Wehner.
\newblock ``Bell nonlocality''.
\newblock
  \href{https://dx.doi.org/https://doi.org/10.1103/RevModPhys.86.419}{Reviews
  of Modern Physics {\bf 86}, 419}~(2014).

\bibitem{sainz2015postquantum}
Ana~Bel{\'e}n Sainz, Nicolas Brunner, Daniel Cavalcanti, Paul Skrzypczyk, and
  Tam{\'a}s V{\'e}rtesi.
\newblock ``Postquantum steering''.
\newblock
  \href{https://dx.doi.org/https://doi.org/10.1103/PhysRevLett.115.190403}{Physical
  review letters {\bf 115}, 190403}~(2015).

\bibitem{gisin1989stochastic}
Nicolas Gisin.
\newblock ``Stochastic quantum dynamics and relativity''.
\newblock
  \href{https://dx.doi.org/http://doi.org/10.5169/seals-116034}{Helvetica
  Physica Acta {\bf 62}, 363--371}~(1989).

\bibitem{hughston1993complete}
Lane~P Hughston, Richard Jozsa, and William~K Wootters.
\newblock ``A complete classification of quantum ensembles having a given
  density matrix''.
\newblock
  \href{https://dx.doi.org/https://doi.org/10.1016/0375-9601(93)90880-9}{Physics
  Letters A {\bf 183}, 14--18}~(1993).

\bibitem{sainz2020bipartite}
Ana~Bel{\'e}n Sainz, Matty~J Hoban, Paul Skrzypczyk, and Leandro Aolita.
\newblock ``Bipartite postquantum steering in generalized scenarios''.
\newblock
  \href{https://dx.doi.org/https://doi.org/10.1103/PhysRevLett.125.050404}{Physical
  Review Letters {\bf 125}, 050404}~(2020).

\bibitem{sainz2018formalism}
Ana~Bel{\'e}n Sainz, Leandro Aolita, Marco Piani, Matty~J Hoban, and Paul
  Skrzypczyk.
\newblock ``A formalism for steering with local quantum measurements''.
\newblock
  \href{https://dx.doi.org/https://doi.org/10.1088/1367-2630/aad8df}{New
  Journal of Physics {\bf 20}, 083040}~(2018).

\bibitem{cavalcanti2022post}
Paulo~J Cavalcanti, John~H Selby, Jamie Sikora, Thomas~D Galley, and
  Ana~Bel{\'e}n Sainz.
\newblock ``Post-quantum steering is a stronger-than-quantum resource for
  information processing''.
\newblock
  \href{https://dx.doi.org/https://doi.org/10.1038/s41534-022-00574-8}{npj
  Quantum Information {\bf 8}, 1--10}~(2022).

\bibitem{sainz2025activation}
Ana~Bel{\'e}n Sainz, Paul Skrzypczyk, and Matty~J Hoban.
\newblock ``Activation of post-quantum steering''.
\newblock \href{https://dx.doi.org/10.1088/1367-2630/ae21fe}{New Journal of
  Physics {\bf 27}, 124508}~(2025).

\bibitem{zjawin2024activation}
Beata Zjawin, Matty~J Hoban, Paul Skrzypczyk, and Ana~Bel{\'e}n Sainz.
\newblock ``Activation of postquantumness in bipartite generalized
  einstein-podolsky-rosen scenarios''.
\newblock
  \href{https://dx.doi.org/https://doi.org/10.1103/PhysRevA.110.042212}{Physical
  Review A {\bf 110}, 042212}~(2024).

\bibitem{hoban2025hierarchy}
Matty~J. Hoban, Tom Drescher, and Ana~Bel{\'{e}}n Sainz.
\newblock ``A hierarchy of semidefinite programs for generalised
  {E}instein-{P}odolsky-{R}osen scenarios''.
\newblock \href{https://dx.doi.org/10.22331/q-2025-01-14-1591}{{Quantum} {\bf
  9}, 1591}~(2025).

\bibitem{cavalcanti2024every}
Paulo~J Cavalcanti, John~H Selby, and Ana~Bel{\'e}n Sainz.
\newblock ``Every nonsignaling channel is common-cause realizable''.
\newblock
  \href{https://dx.doi.org/https://doi.org/10.1103/PhysRevA.109.042211}{Physical
  Review A {\bf 109}, 042211}~(2024).

\bibitem{cavalcanti2015detection}
Daniel Cavalcanti, Paul Skrzypczyk, GH~Aguilar, RV~Nery, PH~Souto Ribeiro, and
  SP~Walborn.
\newblock ``Detection of entanglement in asymmetric quantum networks and
  multipartite quantum steering''.
\newblock \href{https://dx.doi.org/https://doi.org/10.1038/ncomms8941}{Nature
  communications {\bf 6}, 1--6}~(2015).

\bibitem{uola2020quantum}
Roope Uola, Ana~CS Costa, H~Chau Nguyen, and Otfried G{\"u}hne.
\newblock ``Quantum steering''.
\newblock
  \href{https://dx.doi.org/https://doi.org/10.1103/RevModPhys.92.015001}{Reviews
  of Modern Physics {\bf 92}, 015001}~(2020).

\bibitem{jenvcova2022assemblages}
Anna Jen{\v{c}}ov{\'a}.
\newblock ``Assemblages and steering in general probabilistic theories''.
\newblock
  \href{https://dx.doi.org/https://doi.org/10.1088/1751-8121/ac97ce}{Journal of
  Physics A: Mathematical and Theoretical {\bf 55}, 434001}~(2022).

\bibitem{MW}
\url{https://doi.org/10.5281/zenodo.19468921}.

\bibitem{acin2010unified}
A.~Ac\'{\i}n, R.~Augusiak, D.~Cavalcanti, C.~Hadley, J.~K. Korbicz,
  M.~Lewenstein, Ll. Masanes, and M.~Piani.
\newblock ``Unified framework for correlations in terms of local quantum
  observables''.
\newblock \href{https://dx.doi.org/10.1103/PhysRevLett.104.140404}{Phys. Rev.
  Lett. {\bf 104}, 140404}~(2010).

\bibitem{MATLAB}
``Matlab {R2022b} \& {R2025b}''.
\newblock  url:~\url{https://www.mathworks.com/}.

\bibitem{grant2013cvx}
Michael Grant and Stephen Boyd.
\newblock ``{CVX}: {MATLAB} software for disciplined convex programming''.
\newblock  url:~\url{http://cvxr.com/cvx}.

\bibitem{tutuncu1999sdpt3}
K.~C. Toh, M.~J. Todd, and R.~H. Tütüncü.
\newblock ``{SDPT3} — a matlab software package for semidefinite programming,
  version 1.3''.
\newblock
  \href{https://dx.doi.org/https://doi.org/10.1080/10556789908805762}{Optimization
  Methods and Software {\bf 11}, 545--581}~(1999).

\bibitem{qetlab}
Nathaniel Johnston.
\newblock ``{QETLAB}: a {MATLAB} toolbox for quantum entanglement''.
\newblock  url:~\url{http://qetlab.com}.

\bibitem{navascues2007bounding}
Miguel Navascu{\'e}s, Stefano Pironio, and Antonio Ac{\'\i}n.
\newblock ``Bounding the set of quantum correlations''.
\newblock
  \href{https://dx.doi.org/https://doi.org/10.1103/PhysRevLett.98.010401}{Physical
  Review Letters {\bf 98}, 010401}~(2007).

\bibitem{navascues2008convergent}
Miguel Navascu{\'e}s, Stefano Pironio, and Antonio Ac{\'\i}n.
\newblock ``A convergent hierarchy of semidefinite programs characterizing the
  set of quantum correlations''.
\newblock
  \href{https://dx.doi.org/https://doi.org/10.1088/1367-2630/10/7/073013}{New
  Journal of Physics {\bf 10}, 073013}~(2008).

\bibitem{johnston2016extended}
Nathaniel Johnston, Rajat Mittal, Vincent Russo, and John Watrous.
\newblock ``Extended non-local games and monogamy-of-entanglement games''.
\newblock \href{https://dx.doi.org/10.1098/rspa.2016.0003}{Proceedings of the
  Royal Society A: Mathematical, Physical and Engineering Sciences {\bf 472},
  20160003}~(2016).

\bibitem{hoban2018channel}
Matty~J Hoban and Ana~Bel{\'e}n Sainz.
\newblock ``A channel-based framework for steering, non-locality and beyond''.
\newblock
  \href{https://dx.doi.org/https://doi.org/10.1088/1367-2630/aabea8}{New
  Journal of Physics {\bf 20}, 053048}~(2018).

\bibitem{rossi2022characterising}
Vinicius~P Rossi, Matty~J Hoban, and Ana~Bel{\'e}n Sainz.
\newblock ``On characterising assemblages in {Einstein}--{Podolsky}--{Rosen}
  scenarios''.
\newblock
  \href{https://dx.doi.org/https://doi.org/10.1088/1751-8121/ac7090}{Journal of
  Physics A: Mathematical and Theoretical {\bf 55}, 264002}~(2022).

\bibitem{clauser1969proposed}
John~F Clauser, Michael~A Horne, Abner Shimony, and Richard~A Holt.
\newblock ``Proposed experiment to test local hidden-variable theories''.
\newblock
  \href{https://dx.doi.org/https://doi.org/10.1103/PhysRevLett.23.880}{Physical
  review letters {\bf 23}, 880}~(1969).

\end{thebibliography}

\appendix

\section{Classical vs.~quantum steering}\label{app:compare}

In this section we briefly discuss the boundary between classical and quantum steering in our scenarios, and how this differs from the existing literature. 

From a quantum information perspective, steering has been studied within the context of certifying shared entanglement \cite{cavalcanti2015detection,uola2020quantum}. In general, Alice (who is not trusted) wants to convince Bob that they share a system prepared in an entangled state. Now, when considering a multiparty setup, in general there are two ingredients: (i) which parties are trusted (i.e., characterised) and which are not, and (ii) what type of entangled states we wish to certify. For example, one can have one uncharacterised Alice, two characterised Bobs, and wish to certify entanglement across the bipartition `Alice vs.~the Bobs'. Or one could have one uncharacterised Alice, two characterised Bobs, and wish to certify genuinely-tripartite entanglement (see Ref.~\cite[Eq.~(21)]{cavalcanti2015detection}). The crucial thing here then is that the starting point is the assumption of a quantum state $\rho_{AB_1B_2}$, and its entanglement properties are assessed by studying the joint state of the trusted parties $\sigma^{B_1B_2}_{a|x}$ (see Ref.~\cite[Eq.~(100)]{uola2020quantum} and Ref.~\cite[Table II]{cavalcanti2015detection}). Therefore, assemblages here include the information on how the characterised parties correlate.
This in stark contrast with the operational scenario we put forward in this manuscript. As a remark, notice that, in the traditional approach, if an assemblage $\{\sigma^{B_1B_2}_{a|x}\}$ is compatible with a fully separable $\rho_{AB_1B_2}$, then the correlations $p(ab_1b_2|xy_1y_1) = \tr{M^{B_1}_{b_1|y_1} \otimes M^{B_2}_{b_2|y_2} \, \sigma^{B_1B_2}_{a|x}}{}$ will never violate a Bell inequality. 

\bigskip

From a foundational viewpoint, one is interested in the question of when assemblages may admit of a classical explanation. Taking a step back then, given that the parties (trusted and untrusted) are distant agents performing space-like separated actions, one can argue that the only resources that may be shared among them `for free' are classical systems, and quantum operations and manipulation can only happen locally in each laboratory (see Fig.~\ref{f:negmod}). Assemblages that can arise in this manner have been known in the literature as those admitting of a \textit{local hidden state} model (LHS).  If one underpins the study of multipartite steering with this mind frame, then, the question of whether an assemblage is `steerable' indeed asks whether the assemblage cannot be explained via an LHS model. 

Now, in the traditional view of multipartite steering, for a scenario with one Alice and two Bobs, the assemblages are given by the elements $\{\sigma^{B_1B_2}_{a|x}\}$. Since a quantum description is assumed, then an LHS model will have the form: 
\begin{align}\label{eq:tradste}
\sigma^{B_1B_2}_{a|x} = \sum_{\lambda \in \Lambda} p(\lambda) \, p(a|x,\lambda) \, \rho^{B_1}_\lambda \otimes \rho^{B_2}_\lambda \,, 
\end{align}
where $\Lambda$ is the shared classical system,  $p(\lambda)$ is the probability that the system is prepared on state $\lambda$, $p(a|x,\lambda)$ is the local response function of Alice, and $\rho^{B_k}_\lambda$ is the (normalised) local quantum state preparation of the $k$-th Bob conditioned on the value of $\lambda$. Connecting back to entanglement certification, a multipartite assemblage that admits of such an LHS model is compatible with a fully-separable $\rho_{AB_1B_2}$.

\begin{figure}
\begin{center}
\begin{tikzpicture}[scale=0.4]
\node (ai) at (90:5) {} ;
\node (af) at (20:5) {} ;
\node (bi) at (270:5) {} ;
\node (bf) at (160:5) {} ;

\shade[draw, thick, color=gray!40!black ,rounded corners, inner color=white,outer color=gray!50!white] ($ (ai) + (-1,-1) $) rectangle ($ (ai) + (1,1) $) ;
\shade[draw, thick,color=gray!40!black ,rounded corners, inner color=white,outer color=gray!50!white] ($ (af) + (-1,-1) $) rectangle ($ (af) + (1,1) $) ;
\draw[thick, ,rounded corners] ($ (bi) + (-1,-1) $) rectangle ($ (bi) + (1,1) $) ;
\draw[thick, ,rounded corners] ($ (bf) + (-1,-1) $) rectangle ($ (bf) + (1,1) $) ;

\node at ($ (ai)$) {\small{$\mathrm{A}_1$}};
\node at ($ (af)$) {\small{$\mathrm{A}_{n}$}};
\node at ($ (bi)$) {\small{$\mathrm{B}_1$}};
\node at ($ (bf)$) {\small{$\mathrm{B}_N$}};

\node at ($ (ai) + (0,1.8) $) {\small{$p_1(a_1|x_1,\lambda)$}};
\node at ($ (af) + (3.5,-0.1) $) {\small{$p_n(a_n|x_n,\lambda)$}};
\node at ($ (bi) + (0,-1.8) $) {\small{$\rho^{1}_\lambda$}};
\node at ($ (bf) + (-2,0) $) {\small{$\rho^{N}_\lambda$}};

\node[draw,shape=circle,fill,scale=.2] at (45:4) {};
\node[draw,shape=circle,fill,scale=.2] at (55:4) {};
\node[draw,shape=circle,fill,scale=.2] at (65:4) {};
\node[draw,shape=circle,fill,scale=.2] at (210:4) {};
\node[draw,shape=circle,fill,scale=.2] at (200:4) {};
\node[draw,shape=circle,fill,scale=.2] at (220:4) {};

\draw[thick, ->] (0,0) -- ($ (ai) + (0,-1.3) $);
\draw[thick, ->] (0,0) -- ($ (af) + (-1.3,0) $);
\draw[thick, ->] (0,0) -- ($ (bi) + (0,1.3) $);
\draw[thick, ->] (0,0) -- ($ (bf) + (1.3,0) $);

\node[draw, shape=star,star points=7, fill,color=black,inner color=white,scale=1 ] at (0,0) {};
\node at (0,0) {S};

\node at (2.5,0.5) {\tiny{$\lambda$}};
\node at (0.5,2) {\tiny{$\lambda$}};
\node at (-1,-2) {\tiny{$\lambda$}};
\node at (-2.5,0.5) {\tiny{$\lambda$}};
\end{tikzpicture}
\end{center}
\caption{\textbf{Local-Hidden-State model for an assemblage.} 
Steering scenario with $n$ uncharacterised partes (Alices) and $N$ characterised parties (Bobs): A source produces the classical system on state $\lambda$ with  probability $p(\lambda)$ and sends it to the $n+N$ parties. The uncharacterised parties produce their outcomes via the response functions $p_j(a_j|x_j,\lambda)$, whereas the characterised ones locally prepare quantum systems on states $\rho^{k}_\lambda$. In a traditional description of steering, this model prepares an assemblage as per Eq.~\eqref{eq:tradste}. In our operational description of steering, this model explains assemblages as per Eq.~\eqref{eq:LHS}.}
\label{f:negmod}
\end{figure}
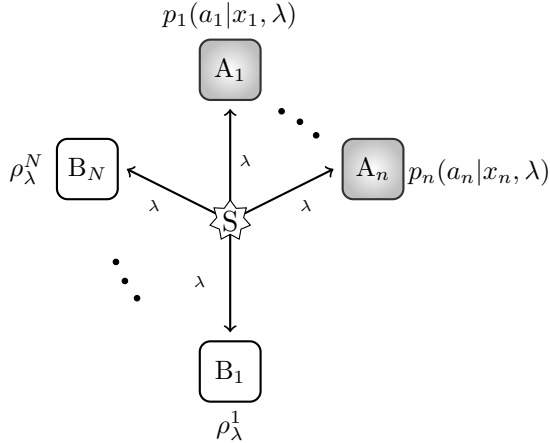

\bigskip  

Now let us consider the operational scenario we have in this manuscript. Here there is no assumption of an underpinning quantum explanation, but rather a description only in terms of the local information each of the parties has  (together with the correlations the Alices observe).   Here, a classical explanation for the assemblage is also given in terms of the LHS model described above. Formally, this reads as follows. 

\begin{defn}\textbf{Classical realisation of an assemblage.--}\\
Given an EPR scenario specified by $\T$, an assemblage $\As_\T$ admits of a classical realisation (a.k.a.~an LHS model) if there exists
\begin{compactenum}
\item a classical system\footnote{In general this can be an infinite-dimensional system, in which case the sum in Eq.~\ref{eq:LHS} turns into an integral. } $\Lambda$, 
\item a probability distribution $p(\lambda)$ for $\lambda \in \Lambda$, 
\item a response function (normalised conditional probability distribution)  $p_j(a_j|x_j,\lambda)$ for each Alice ($j=1:n$), and 
\item a local (normalised) quantum state $\rho^{k}_\lambda$ for each Bob ($k=1:N$) 
\end{compactenum}
such that: 
\begin{align}\label{eq:LHS}
\sigma^{(k)}_{\mathbf{a}|\mathbf{x}} = \sum_{\lambda \in \Lambda} p(\lambda) \, \rho^{k}_\lambda \,  \prod_{j=1:n} \, p_j(a_j|x_j,\lambda)  \,, \quad \forall \, k,\mathbf{a},\mathbf{x}\,.
\end{align}
This classical model is depicted in Fig.~\ref{f:negmod}. 
\end{defn}

\begin{rem}
As a first remark, notice that an assemblage in our steering scenario admits of an LHS model if it also admits of an LHS model when described in the traditional framing of Eq.~\eqref{eq:tradste}.
\end{rem}

As we see next, however, admitting of an LHS model in our operational scenario does not imply that the correlations $p(\mathbf{a}\mathbf{b}|\mathbf{x}\mathbf{y})$ cannot violate a Bell inequality. 

\begin{prop}
There exists assemblages $\As_\T$ that admit of an LHS model  but that, when explained by quantum theory and for  some local measurements by the Bobs, the correlations $p(\mathbf{a}\mathbf{b}|\mathbf{x}\mathbf{y})$ violate a Bell inequality. 
\end{prop}
\begin{proof}
We prove this proposition by finding an example. Consider a steering scenario $\T = \left(n=1,|\X|=2,|\A|=2,N=2,d_1=2,d_2=2 \right)$. Now consider the assemblage $\As_\T$ that arises from Alice performing measurements 
$\{M_{0|0} = \ket{+}\bra{+}, M_{1|0} = \ket{-}\bra{-}\}$ and $\{M_{0|1} = \ket{0}\bra{0}, M_{1|1}=\ket{1}\bra{1}\}$ 
on the three-qubit system prepared in the state
$\frac{\id}{2} \otimes \ket{\Psi^+}\bra{\Psi^+}$, where $\ket{\Psi} = \frac{\ket{01}+\ket{10}}{\sqrt{2}}$. This assemblage has elements
\begin{align*}
\As_\T \,: \quad 
 \sigma^{(k)}_{a|x} =  \frac{1}{4}
\begin{bmatrix}
    1    &     0 \\
         0    &     1
\end{bmatrix} 
\,,\quad \forall \, a,x \quad k=1,2\,.
\end{align*}
On the one hand, $\As_\T$ admits of a simple LHS model where we don't even need to share a classical system $\Lambda$ to correlate the parties: Alice flips an unbiased coin regardless of $x$, rendering $p(a|x)=\tfrac{1}{2}$, and each Bob prepares locally a maximally mixed state $\rho^k = \frac{\id}{2}$. 

On the other hand, consider the correlations $p(\mathbf{a}\mathbf{b}|\mathbf{x}\mathbf{y})$ that arise when the first Bob measures 
\begin{align*}
M^{B_1}_{b|0} = \frac{\id + (-1)^b X}{2} \,,\quad M^{B_1}_{b|1} = \frac{\id + (-1)^b Z}{2}\,,
\end{align*}
and the second Bob measures
\begin{align*}
M^{B_2}_{b|0} = \frac{\id + (-1)^b (X+Z)}{2} \,,\quad M^{B_2}_{b|1} = \frac{\id + (-1)^b (X-Z)}{2}\,,
\end{align*}
where recall that $X$ and $Z$ are Pauli matrices. 

These correlations have the form $p(\mathbf{a}\mathbf{b}|\mathbf{x}\mathbf{y}) = \frac{1}{2} p(\mathbf{b}|\mathbf{y})$, where  $p(\mathbf{b}|\mathbf{y})=p(b_1b_2|y_1y_2)$  violates the CHSH inequality maximally \cite{clauser1969proposed}. Hence, $p(\mathbf{a}\mathbf{b}|\mathbf{x}\mathbf{y})$ cannot admit of a classical model. 
\end{proof}

\bigskip

We see then that in our theory-independent operational approach the relationship between nonclassicality in steering and in Bell scenarios is less straightforward than in scenarios with only one characterised party. Here we see that one can have classically-realisable assemblages that yield nonclassical correlations when the Bobs perform local measurements. This parallels the situation we found in Sec.~\ref{se:pqnwpqs}, where assemblages that admit of a quantum realisation may yield postquantum correlations in the corresponding Bell scenario.  We see then that the phenomena of steering and Bell nonlocality do not form a strict hierarchy, as usually thought.  


\section{Numerical tolerances}\label{ap:A}

A natural question in postquantum certification via Thm.~\ref{thm:hasmarginal} is how much should one relax the constraints of the semidefinite program in order to accommodate a solution. One possibility is to allow the matrices $\{\sigma_{\mathbf{a}|\mathbf{x}}\}$ to be slightly negative. Denote by $\lambda$ the smallest of all the eigenvalues of all the matrices $\{\sigma_{\mathbf{a}|\mathbf{x}}\}$  (recall these matrices are the variables in our optimisation problem).  One can then ask what the  least non-positive  value of $\lambda$ is which can accommodate a solution for the SDP of Thm.~\ref{thm:hasmarginal} for a given assemblage.  For instance, if a solution to the optimisation program exists with $\lambda=0$ then the assemblage actually admits of parent sates $\{\sigma_{\mathbf{a}|\mathbf{x}}\}$. If a solution only exists with $\lambda < 0$ then parent states do not exist.  
 Notice that $\lambda$ should not be thought of as  a measure of postquantum steering, but rather as a way to assess the numerical robustness of a given assemblage. In the following table, we include the value of $\lambda$ for each of the numerical examples included in this work. 

\begin{center}
\begin{tabular}{|c|c|}
\hline
Example & $\lambda$ \\ 
\hline 
ABB\_1 & $-0.00185$ \\
ABB\_2 & $-0.000157$ \\
ABB\_PTP\_1 & $-0.000104$ \\
ABB\_PTP\_2 & $-0.000085$ \\
AABB\_PTP\_1 & $-0.00027$\\
AABB\_PTP\_2 & $-0.00019$\\
\hline
\end{tabular}
\end{center}
We see that the required negativity is well above the numerical precision, which shows our examples are not mere numerical instabilities.

\end{document}